\documentclass[10pt,twocolumn,conference]{IEEEtran}
\usepackage{times}
\usepackage[reqno]{amsmath}
\usepackage{amsfonts}
\usepackage{caption3}
\usepackage{times,amsmath,epsfig}
\usepackage{latexsym,amssymb}
\usepackage{rotating,color}
\usepackage{cite}
\usepackage{extarrows}
\usepackage{amsthm}
\usepackage{lipsum}
\usepackage{multicol}

\def\myQED{\mbox{\rule[0pt]{1.5ex}{1.5ex}}}

\renewcommand{\arg}{{\hbox{argmin}}}
\newcommand{\argmin}{\mathop{\mathrm{argmin}}}

\newcommand{\no}{\nonumber}

\newtheorem{thm}{Theorem}

\newtheorem{prop}{Proposition}
\newtheorem{lem}{Lemma}
\newtheorem{cor}{Corollary}

\newtheorem{rmk}{Remark}

\IEEEoverridecommandlockouts

\makeatletter
\def\ps@headings{%
\def\@oddhead{\mbox{}\scriptsize\rightmark \hfil \thepage}%
\def\@evenhead{\scriptsize\thepage \hfil \leftmark\mbox{}}%
\def\@oddfoot{}%
\def\@evenfoot{}}
\makeatother
\begin{document}
\title{Bayesian Quickest Change-Point Detection with Sampling Right Constraints}
\author{Jun Geng, \emph{Student Member, IEEE}, Erhan Bayraktar and Lifeng Lai, \emph{Member, IEEE} \thanks{The work of J. Geng and L. Lai was supported by the National Science Foundation under grant DMS-12-65663. The work of E. Bayraktar was supported by National Science Foundation under grant DMS-11-18673. This paper was presented in part at Annual Allerton Conference on Communication, Control and Computing, Monticello, IL, Oct. 2012, and at IEEE International Conference on Acoustics, Speech, and Signal Processing (ICASSP), Florence, Italy, May 2014.

J. Geng and L. Lai are with the Department of Electrical and Computer Engineering, Worcester Polytechnic Institute, Worcester, MA 01609, USA (Emails: \{jgeng, llai\}@wpi.edu).

E. Bayraktar is with the Department of Mathematics, University of Michigan, Ann Arbor, MI 48109, USA (Email:erhan@umich.edu).}}
\maketitle 



\begin{abstract}
In this paper, Bayesian quickest change detection problems with sampling right constraints are considered. Specifically, there is a sequence of random variables whose probability density function will change at an unknown time. The goal is to detect this change in a way such that a linear combination of the average detection delay and the false alarm probability is minimized. Two types of sampling right constrains are discussed. The first one is a limited sampling right constraint, in which the observer can take at most $N$ observations from this random sequence. Under this setup, we show that the cost function can be written as a set of iterative functions, which can be solved by Markov optimal stopping theory. The optimal stopping rule is shown to be a threshold rule. An asymptotic upper bound of the average detection delay is developed as the false alarm probability goes to zero. This upper bound indicates that the performance of the limited sampling right problem is close to that of the classic Bayesian quickest detection for several scenarios of practical interest. The second constraint discussed in this paper is a stochastic sampling right constraint, in which sampling rights are consumed by taking observations and are replenished randomly. The observer cannot take observations if there are no sampling rights left. We characterize the optimal solution, which has a very complex structure. For practical applications, we propose a low complexity algorithm, in which the sampling rule is to take observations as long as the observer has sampling rights left and the detection scheme is a threshold rule. We show that this low complexity scheme is first order asymptotically optimal as the false alarm probability goes to zero.
\end{abstract}



\section{Introduction} \label{sec:intro}
Quickest change-point detection aims to detect an abrupt change in the probability distribution of a stochastic process with a minimal detection delay. Bayesian quickest detection \cite{Shiryaev:Soviet:61, Shiryaev:TPIA:63} is one of the most important formulations. In the classic Bayesian setup, there is a sequence of random variables $\{X_{n}, n=1,2,\ldots\}$ with a geometrically distributed change-point $\Lambda$. Before the change-point $\Lambda$, the sequence $X_{1}, \ldots, X_{\Lambda-1}$ is assumed to be independent and identically distributed (i.i.d.) with probability density function (pdf) $f_0(x)$, and after $\Lambda$, the sequence is assumed to be i.i.d. with pdf $f_1(x)$. The goal is to find an optimal stopping time $\tau$, at which the change is declared, that minimizes the detection delay under a false alarm constraint.

In recent years, this technique has found a lot of applications in wireless sensor networks~\cite{Mei:TIT:05,Moustakides:ICIF:06,Tartakovsky:SA:08,Tartakovsky:ICIF:03,Tartakovsky:ICIF:06, Fellouris:Bernoulli:13, Hadjiliadis:CDC:09} for network intrusion detection \cite{Premkumar:INFOC:08}, seismic sensing~\cite{Pisarenko:PEPI:87}, structural health monitoring, etc. In such applications, sensors are deployed to monitor their surrounding environment for abnormalities. Such abnormalities, which are modeled as change-points, typically imply certain activities of interest. For example, a sensor network may be built into a bridge to monitor its structural health condition. In this case, a change may imply that a certain structural problem, such as an inner crack, has occurred in the bridge. In this context, the false alarm probability and the detection delay between the time when a structural problem occurs and the time when an alarm is raised are of interest.

In the classic quickest change detection setups, one can observe the underlying signal at each time slot. In the above mentioned applications, however, the situation is different. Taking samples and computing statistics cost energy. Sensors are typically powered by batteries with limited capacity and/or are charged randomly with renewable energy. Hence in these applications, it is unlikely that one can take samples at all time slots. For example, for sensors powered by a battery, they are subjected to a limited energy constraint. Hence, they have only limited energy to make a fixed number of observations. For sensors powered by renewable energy, they are subjected to a stochastic energy constraint. The sensors cannot take observations unless there are energy left in the battery.

In this paper, motivated by above applications, we extend the classic Bayesian quickest change-point detection by imposing casual energy constraints. Specifically, we relax the assumption in the classic Bayesian setup that the observer can observe the underlying signal freely at any time slots. Instead, we assume that an observation can be taken only if the sensor has energy left in its battery. The sensor has the freedom to choose the sampling time, but it has to plan its use of energy carefully due to the energy constraint. The goal of the sensor is to find the optimal sampling strategy (or the optimal energy utility strategy) and the optimal stopping rule to minimize the average detection delay under a false alarm constraint. The optimal solutions of the proposed problems are obtained by dynamic programming (DP). However, the optimal solutions in general do not have a close form expression due to the iterative nature of DP. Although the optimal solutions can be solved numerically, numerical method provides us little insight of the optimal solutions. Hence, in this paper, we also conduct asymptotic analysis and design low-complexity asymptotically optimal schemes.

In particular, we consider two types of constraints in this paper. The first one is a limited observation constraint. Specifically, the sensor is allowed to take at most $N$ observations. After taking each observation, the sensor needs to decide whether to stop and declare a change, or to continue sampling. If the sensor decides to continue, it also then needs to determine the next sampling time. In this paper, we develop the optimal stopping rule and the sampling rule for this problem. The optimal stopping rule is shown to be a threshold rule, and the optimal sampling time of the $n^{th}$ observation is the one minimizing the most updated cost function. An asymptotic upper bound of the average detection delay is developed as the false alarm probability goes to zero. The derived upper bound indicates that the average detection delay is close to that of the setup without energy constraint \cite{Tartakovsky:TPIA:04} when $N$ is sufficiently large or when $f_{0}$ and $f_{1}$ are close to each other.

The second constraint being considered is a stochastic energy constraint. This constraint is designed for sensors powered by renewable energy. In this case, the energy stored in the sensor is consumed by taking observations and is replenished by a random process. The sensor cannot store extra energy if its battery is full, and the sensor cannot take observations if its battery is empty. Hence, the sensor needs to find a strategy to use its energy efficiently. Under this constraint, we develop the optimal stopping rule and the optimal sampling rule. The complexity of the optimal solution, however, is very high. To address this issue, we design a low complexity algorithm in which the sensor takes observations as long as there is energy left in its battery and the sensor detects the change by using a threshold rule. We show that this simple algorithm is first order asymptotically optimal as the false alarm probability goes to zero.

Although these problem formulations are originally motivated by wireless sensor networks, their applications are not limited to this area. For example, in clinical trials, it is desirable to quickly and accurately obtain the efficiency of certain medicine or therapy with limited number of tests, since it might be very costly and sometime even health-damaging to conduct a test. Hence, the limited observation constraint can be applied in this scenario. Therefore, in the remainder of this paper, instead of using application specific concepts such as ``sensor'' and ``energy constraint'', we use general terms such as ``observer'' and ``sampling right constraint''.

The problems considered in this paper are related to recent works on the quickest change-point detection problem that take the observation cost into consideration. In particular, \cite{Dzhamburia:Steklov:83} assumes that each observation is worth either $1$ if it is observed or $0$ if it is skipped. \cite{Dzhamburia:Steklov:83} is interested in minimizing both the Bayesian detection delay and the total cost made by taking observations. Moreover, \cite{Dzhamburia:Steklov:83} considers both discrete and continuous time case and shows the existence of the optimal stopping rule-sampling strategy pair.   \cite{Bayraktar:12:XX}, which considers the Bayesian quickest change-point detection problem with sampling right constraints in the continuous time scenario, is also relevant to our paper. \cite{Bayraktar:12:XX} considers two cases: the observer has a fixed sampling rights or the observer's sampling rights arrive according to a Poisson process. \cite{Bayraktar:12:XX} characterizes the optimal solution for these problems. Compared with~\cite{Dzhamburia:Steklov:83, Bayraktar:12:XX}, our paper focuses the discrete time case, and provides low complexity asymptotically optimal solutions as well as optimal solutions.

We also briefly mention other related papers. The first main line of existing works considers the problem under a Bayesian setup. In particular, \cite{Premkumar:INFOC:08} considers a wireless network with multiple sensors monitoring the Bayesian change in the environment. Based on the observations from sensors at each time slot, the fusion center decides how many sensors should be activated in the next time slot to save energy.~\cite{Banerjee:SADA:12} takes the average number of observations taken before the change-point into consideration, and it provides the optimal solution along with low-complexity but asymptotically optimal rules. \cite{Tartakovsky:SADA:10} is a recent comprehensive survey that summarizes the current development on the Bayesian quickest change-point detection problem. There are also some existing works consider the problem under minmax setting. For example, \cite{Jun:TSP:13} considers the non-Bayesian quickest detection with a stochastic sampling right constraint. \cite{Banerjee:TIT:12, Banerjee:ISIT:13} extend the constraint of the average number of observations into non-Bayesian setups and sensor networks. \cite{Polunchenko:MCAP:12} is a recent survey on the quickest change-point detection problem which comprehensively summarizes the progress made on both Bayesian and non-Bayesian setups. 

The remainder of this paper is organized as follows. Our mathematical model for the Bayesian quickest change-point detection problem with sampling right constraints is described in Section \ref{sec:model}. Section \ref{sec:lmt_opt} presents the optimal solution and the asymptotic upper bound for the limited sampling right problem. Section~\ref{sec:stc_opt} provides the optimal and the asymptotically optimal solution for the stochastic sampling right problem. Numerical examples are given in Section~\ref{sec:simulation}. Finally, Section~\ref{sec:conclusion} offers concluding remarks.

\section{Model}\label{sec:model}
Let $\{X_{k}, k=1, 2, \ldots\}$ be a sequence of random variables with an unknown change-point $\Lambda$. $\{X_{k}\}$'s are i.i.d. with pdf $f_0(x)$ before the change-point $\Lambda$, and i.i.d. with pdf $f_1(x)$ after $\Lambda$. The change-point $\Lambda$ is modeled as a geometric random variable with parameter $\rho$, i.e., for $0 < \rho <1$, $0 \leq \pi < 1$,
\begin{eqnarray}\label{eq:prior}
P(\Lambda = \lambda) = \left\{ \begin{array}{cc}
                               \pi & \lambda = 0 \\
                               (1-\pi)\rho(1-\rho)^{\lambda-1} & \lambda = 1, 2, \ldots
                               \end{array}. \right.    \label{eq:geometry}
\end{eqnarray}
We use $P_{\pi}$ to denote the probability measure under which $\Lambda$ has the above distribution. We will denote the expectation under this measure by $\mathbb{E}_{\pi}$. Additionally, we will use $P_{\lambda}$ and $\mathbb{E}_{\lambda}$ to denote the probability measure and the expectation under the event $\{ \Lambda=\lambda \}$.

We assume that the observer initially has $N$ sampling rights, and her sampling rights are consumed when she takes observations and are replenished randomly. The sampling right replenishing procedure is modeled as a stochastic process $\nu = \{\nu_{1}, \nu_{2}, \dots, \nu_{k}, \dots \}$, where $\nu_{k}$ is the amount of sampling rights collected by the observer at time slot $k$. Specially, $\nu_{k} \in \mathcal{V}=\{0, 1, 2, \ldots \}$, in which $\{ \nu_{k} = 0\}$ implies that she obtains no sampling right at time slot $k$ and $\{ \nu_{k} = i\}$ implies that she collects $i$ sampling rights at $k$. We use $p_{i}=P^{\nu}(\nu_{k} = i)$ to denote its probability mass function (pmf). We assume that $\{\nu_{k} \}$ is i.i.d. over $k$.

The observer can decide when to spend her sampling rights to take observations. Let $\mu = \{\mu_{1}, \mu_{2}, \dots, \mu_{k}, \dots \}$ be the sampling strategy with $\mu_{k} \in \{ 0, 1 \}$, in which $\{\mu_{k} = 1\}$ means that she spends one sampling right on taking observation at time slot $k$ and $\{\mu_{k} = 0\}$ means that no sampling right is spent at $k$ and hence no observation is taken.

We are interested in the case that the observer has a finite sampling right capacity $C$. Let $N_{k}$ be the amount of sampling rights at the end of time slot $k$. $N_{k}$ evolves according to
\begin{eqnarray}
N_{k} = \min\{ C, N_{k-1} + \nu_{k} - \mu_{k}\} \label{eq:N_evolve}
\end{eqnarray}
with $N_{0}=N$. The observer's strategy $\mu$ must obey a causality constraint: \emph{the observer cannot take an observation at time slot $k$ if she has no sampling right at that time slot.} Hence, the admissible strategy set can be written as
\begin{eqnarray}
\mathcal{U} = \left\{ \mu: N_{k}\geq 0, \quad k = 1,2,\ldots. \right\}. \label{eq:causality1}
\end{eqnarray}

The observer spends sampling rights to take observations. We denote the observation sequence as $\left\{ Z_{k}, k=1,2,\ldots \right\}$ with
\begin{eqnarray}
Z_{k}= \left\{ \begin{array}{ll}
X_{k} & \textrm{if  } \mu_{k}=1 \\
\phi & \textrm{if  } \mu_{k}=0
\end{array} \right. ,          \no
\end{eqnarray}
in which $\phi$ denotes no observation.

We call an observation $Z_{k}$ a non-trivial observation if $\mu_{k}=1$, i.e., if the observation is taken from the environment. Denote $t_{i}$ as the time instance that the observer makes the $i^{th}$ observation, then $\mu_{t_{i}} = 1$ and the non-trivial observation sequence can be denoted as $\{X_{t_{1}}, X_{t_{2}}, \ldots, X_{t_{n}}, \ldots\}$.

The observation sequence $\{Z_{k}\}$ generates the filtration $\{\mathcal{F}_{k}\}_{k\in\mathbb{N}}$ with
\begin{eqnarray}
\mathcal{F}_k=\sigma(Z_1,\cdots,Z_k, \{ \Lambda = 0 \}), k=1, 2, \ldots. \no
\end{eqnarray}
and $\mathcal{F}_0$ contains the sample space $\Omega$ and $\{ \Lambda = 0 \}$.

\begin{figure}[thb]
\centering
\includegraphics[width=0.45 \textwidth]{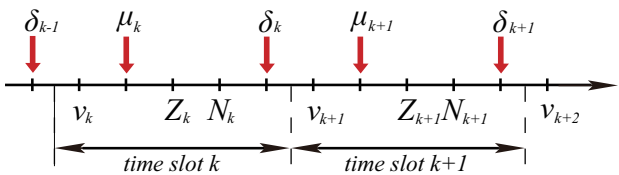}
\caption{The observer's decision flow}
\label{fig:sys_model}
\end{figure}

Figure \ref{fig:sys_model} illustrates the observer's decision flow. At each time slot $k$, the observer has to make two decisions: the sampling decision $\mu_{k}$ and the terminal decision $\delta_{k} \in \{ 0, 1\}$. These two decisions are based on different information. First, the observer needs to decide whether she should spend a sampling right to take an observation ($\mu_{k}=1$) or not ($\mu_{k}=0$) after she obtains the information of $\nu_{k}$. In general, $\mu_{k}$ depends casually on the observation process, the sampling strategy and the sampling right replenishing process, i.e.,
\begin{eqnarray}
\mu_{k} = g_{k}(\mathbf{Z}_{1}^{k-1}, \nu_{1}^{k}, \mu_{1}^{k-1}), \label{eq:muk}
\end{eqnarray}
in which $\mathbf{Z}_{1}^{k-1}$ denotes $\{Z_{1}, \ldots, Z_{k-1}\}$, $\nu_{1}^{k}$ and $\mu_{1}^{k-1}$ are defined in a similar manner, and $g_{k}$ is the sampling strategy function used at $k$. After making each observation $Z_{k}$ (whether it is a non-trival observation in the case of $\mu_k=1$ or it is a trivial observation in the case of $\mu_k=0$), the observer needs to decide whether she should stop sampling and declare that a change has occurred ($\delta_{k}=1$), or to continue the sampling procedure ($\delta_{k}=0$). Therefore, $\delta_{k}$ is a $\mathcal{F}_{k}$ measurable function. We introduce a random variable $\tau$ to denote the time when the observer decides to stop, i.e., $\{ \tau = k\}$ if and only if $\{ \delta_{k} = 1\}$, then $\tau$ is a stopping time with respect to the filtration $\{\mathcal{F}_{k}\}$.

We notice that the distribution of $Z_{k}$ is related to both $X_{k}$ and $\mu_{k}$. Unlike the classic Bayesian setup which only takes the expectation with respect to $P_{\pi}$, in our setup we should take the expectation with respect to both $P_{\pi}$ and $P^{\nu}$. Hence, we use the superscript $\nu$ over the probability measure and the expectation to emphasize that we are working with a probability measure taken the distribution of the process $\nu$ into consideration. Specifically, we use $P_{\pi}^{\nu}$ and $\mathbb{E}_{\pi}^{\nu}$ to denote the probability measure and the expectation under $\Lambda$, respectively; and we use $P_{\lambda}^{\nu}$ and $\mathbb{E}_{\lambda}^{\nu}$ under the event $\{\Lambda = \lambda \}$.

In this paper, our goal is to design a strategy pair $(\tau, \mu)$ to minimize the detection delay subject to a false alarm constraint. In particular, the average detection delay (ADD) is defined as
$$ \mathrm{ADD}(\pi, N, \tau, \mu) = \mathbb{E}_{\pi}^{\nu}\left[ (\tau-\Lambda)^{+}\right],$$
where $x^{+}=\max\{0, x\}$, and the probability of the false alarm (PFA) is defined as
$$ \mathrm{PFA}(\pi, N, \tau, \mu) = P_{\pi}^{\nu}(\tau<\Lambda).$$
With the initial probability $\pi_{0} = \pi$ and the initial sampling right $N_{0}=N$, we want to solve the following optimization problem:
\begin{eqnarray}
\text{ (P1) } &&\min_{\mu \in \mathcal{U}, \tau \in \mathcal{T}} \mathrm{ADD}(\pi, N, \tau, \mu) \no\\ &&\text{ subject to } \mathrm{PFA}(\pi, N, \tau, \mu) \leq \alpha. \no
\end{eqnarray}
in which $\mathcal{T}$ is the set of all stopping times with respect to the filtration $\{\mathcal{F}_{k}\}$ and $\alpha$ is the false alarm level. By Lagrangian multiplier, for each $\alpha$ the optimization problem (P1) can be equivalently written as
\begin{eqnarray}
\text{ (P2) } J(\pi, N) = \inf_{\mu \in \mathcal{U}, \tau \in \mathcal{T}} U(\pi, N, \tau, \mu), \no 
\end{eqnarray}
where
\begin{eqnarray}
U(\pi, N, \tau, \mu) \triangleq \mathbb{E}_{\pi}^{\nu}\left[ c(\tau-\Lambda)^{+} + \mathbf{1}_{\{ \tau < \Lambda \}} \right] \label{eq:risk}
\end{eqnarray}
for an appropriately chosen constant $c$. We would like to characterize $J(\pi, N)$ in this paper.


\section{Problems with the Limited Sampling Right Constraint}  \label{sec:lmt_opt}
We first consider a special case that $p_{0}=P^{\nu}(\nu_{k}=0) = 1$, that is, other than the initial sampling rights, there will be no additional sampling rights arriving at the observer. Hence she can take at most $N_{0}=N$ observations from the sequence $\{ X_{k} \}$ for the detection purpose. Therefore, we name the sampling right causality constraint as a limited sampling right constraint in this case.

From \eqref{eq:N_evolve} and \eqref{eq:causality1}, it is easy to verify that there are at most $N$ nonzero elements in $\mu$. Hence, instead of considering $\mu = \{ \mu_{k} \}$ with infinite elements, we can describe the sampling strategy by the sampling time sequence $\mu = \{ t_{1}, \ldots, t_{\eta}\}$, where $t_{\eta}$ is the time instance that the observer takes the last observation, and $\eta$ is the number of observations taken by the observer when she stops. Hence, in this paper we term $\eta$ as the sample size, and we notice that $\eta$ is a random variable whose realization varies from different trials. The admissible strategy set \eqref{eq:causality1} can be equivalently written as $\mathcal{U}_N=\{\mu: \eta\leq N\}$ in this case.

In addition, as indicated in Section \ref{sec:model}, in general we need to take the expectation with respect to both $P_{\pi}$ and $P^{\nu}$. However, in this special case we only need to take expectation with respect to $P_{\pi}$ since the process $\nu$ has no randomness. Therefore,  $\mathbb{E}^{\nu}_{\pi}$ and $P^{\nu}_{\pi}$ can be replaced by $\mathbb{E}_{\pi}$ and $P_{\pi}$ respectively. In particular, the cost function can be written as
\begin{eqnarray}
U(\pi, N, \tau, \mu) = \mathbb{E}_{\pi}\left[ c(\tau-\Lambda)^{+} + \mathbf{1}_{\{ \tau < \Lambda \}} \right].
\end{eqnarray}

\subsection{Optimal Solution}\label{sec:opt}
Let $\pi_k$ be the posterior probability that a change has occurred at the $k^{th}$ time instance, namely
\begin{eqnarray}
\pi_k=P(\Lambda \leq k|\mathcal{F}_k), \quad k=0,1,\ldots. \label{eq:posterior}
\end{eqnarray}
Using Bayes' rule, $\pi_{k}$ can be shown to satisfy the recursion
\begin{eqnarray}\label{eq:post}
\pi_k=\left\{\begin{array}{ll}\Phi_0(\pi_{k-1}),& \text{if } \mu_{k} = 0\\
\Phi_1(X_{k}, \pi_{k-1}),&  \text{if } \mu_{k} = 1 \end{array} \right., \label{eq:recursion}
\end{eqnarray}
in which
\begin{eqnarray}
\Phi_0(\pi_{k-1}) = \pi_{k-1}+(1-\pi_{k-1})\rho, \label{eq:kernel0}
\end{eqnarray}
and
\begin{eqnarray} 
&&\hspace{-10mm}\Phi_1(X_{k}, \pi_{k-1}) \no\\
&&= \frac{\Phi_0(\pi_{k-1})f_1(X_k)}{\Phi_0(\pi_{k-1})f_1(X_k)+(1-\Phi_0(\pi_{k-1}))f_0(X_k)}. \label{eq:kernel1}
\end{eqnarray}

It turns out that $\pi_{k}$ is a sufficient statistic for this problem, as the next result demonstrates.
\begin{prop} \label{prop:cost}
For each sampling strategy $\mu$ and stopping rule $\tau$
\begin{eqnarray}
U(\pi, N, \tau, \mu)=\mathbb{E}_{\pi}\left[1-\pi_{\tau}+c\sum\limits_{k=0}^{\tau-1}\pi_k\right]. \label{eq:cost}
\end{eqnarray}
\end{prop}
\begin{proof}
An outline of the proof is provided as follows:
\begin{eqnarray}
U(\pi, N, \tau, \mu) &=& \mathbb{E}_{\pi}\left[c(\tau-\Lambda)^{+}+\mathbf{1}_{\{\tau<\Lambda\}} \right] \no \\
&=&\mathbb{E}_{\pi}\left[c(\tau-\Lambda)\mathbf{1}_{\{\tau\geq\Lambda\}} + \mathbf{1}_{\{\tau<\Lambda\}}\right] \no \\
&=&\mathbb{E}_{\pi}\left[c\sum_{k=0}^{\tau-1}\mathbf{1}_{\{\Lambda \leq k\}} + \mathbf{1}_{\{\tau<\Lambda\}}\right] \no \\
&=&\mathbb{E}_{\pi}\left[c\sum_{k=0}^{\tau-1}\pi_{k} + (1-\pi_{\tau})\right]. \no
\end{eqnarray}
A detailed proof follows closely to that of Proposition 5.1 of~\cite{Poor:Book:08} and is omitted for brevity.
\end{proof}

We first have the following lemma characterizing some properties of the optimal $(\tau, \mu)$:
\begin{lem} \label{lem:ext}
Let $\mu=\{t_{1}, \ldots, t_{\eta}\}$ be an admissible sampling strategy, and $\tau$ be a stopping time. If $\eta < N$ and $\tau > t_{\eta}$, then $(\tau, \mu)$ is not optimal.\end{lem}
\begin{proof}
The proof is provided in Appendix \ref{apd:ext}.
\end{proof}

This result implies that if the observer has any sampling rights left, it is not optimal for him to stop at time slot $k$ without taking an observation at $k$. In other words, the only scenario in which the observer may stop sometime after an observation is taken occurs when she has exhausted all her sampling rights. From this lemma, we immediately have the following result.
\begin{cor}\label{cor:ext}
If $\mu^{*} = \{t_{1}^{*}, \ldots, t_{\eta^{*}}^{*}\}$ is the optimal sampling strategy, then on the event $\{\eta^{*} < N\}$, we have $\tau^{*} = t_{\eta^{*}}^{*}$.
\end{cor}

We solve (P2) by using the dynamic programming principle. Similar to the approach used in~\cite{Bayraktar:SPA:09}, we define a functional operator $\mathcal{G}$ as
\begin{eqnarray}
\mathcal{G}V(\pi) = \min \left\{1-\pi, \inf_{m \geq 1}\mathbb{E}_{\pi}\left[c\sum_{k=0}^{m-1}\pi_k+ V(\pi_m)\right]\right\}, \label{eq:operator}
\end{eqnarray}
in which
\begin{eqnarray}
\pi_0 &=& \pi, \no \\
\pi_{k} &=& \pi+\sum\limits_{i=1}^k(1-\pi)\rho(1-\rho)^{i-1}, \quad k=1,\cdots m-1,\no\\
\pi_{m}&=&\frac{\Phi_0(\pi_{m-1})f_1(X_m)}{\Phi_0(\pi_{m-1})f_1(X_m)+(1-\Phi_0(\pi_{m-1}))f_0(X_m)}.\no
\end{eqnarray}
Using this functional operator, we can introduce a set of iteratively defined functions:
\begin{eqnarray}
V_0(\pi)&=&\min\limits_{m\geq 0}\left[c\sum\limits_{k=0}^{m-1}\pi_k+1-\pi_m\right], \label{eq:V0} \\
V_n(\pi) &=& \mathcal{G}V_{n-1}(\pi), \quad n = 1, \ldots, N. \label{eq:iteration}
\end{eqnarray}

The operator $\mathcal{G}$ converts (P2) to a Markov stopping problem. Specifically, we have the following result:
\begin{thm}\label{thm:optimal}
For all $n=0,\cdots, N$, $\pi_{0} = \pi \in [0, 1)$, we have
\begin{eqnarray}
J(\pi, n)=V_n(\pi). \no
\end{eqnarray}
Furthermore, by letting $t_{0}^{*}=0$, the optimal sampling time for (P2) can be determined by
\begin{eqnarray}\label{eq:time1}
t_{n+1}^{*}-t_n^{*} =\argmin\limits_{m\geq 1}\mathbb{E}_{\pi_{t^{*}_n}}\left[c\sum_{k=0}^{m-1}\pi_k+ V_{N-n-1}(\pi_{m})\right],
\end{eqnarray}
for $n=0, 1, \ldots, N-1$. The optimal sampling size is given as
\begin{eqnarray}\label{eq:stop}
\eta^{*} &=& \inf\left\{0 \leq n \leq N: \pi_{t^{*}_n} \in \mathcal{S}_{n} \right\},
\end{eqnarray}
in which $\mathcal{S}_{n}$ is the stopping domain defined as
\begin{eqnarray}
\mathcal{S}_{n} \hspace{-1mm} \triangleq \hspace{-1mm} \left\{\pi_{t_n}: 1-\pi_{t_n} \hspace{-1mm} \leq \hspace{-1mm} \inf_{m \geq 1}\mathbb{E}_{\pi_{t_{n}}}\left[c\sum_{k=0}^{m-1}\pi_k+ V_{N-n-1}(\pi_{m})\right] \hspace{-1mm} \right\},\no
\end{eqnarray}
for $n=0,\cdots, N-1$, and $\mathcal{S}_{N}\triangleq[0, 1]$. In addition, the optimal stopping time is given as
\begin{eqnarray}
\tau^{*} &=& t_{\eta^{*}}^{*} + m^{*}\mathbf{1}_{\{\eta^{*} = N\}}, \label{eq:tau}
\end{eqnarray}
where
\begin{eqnarray}
m^{*} = \argmin_{m \geq 0} \mathbb{E}_{\pi_{t^{*}_{N}}}\left[c\sum\limits_{k=0}^{m-1}\pi_k+1-\pi_m\right]. \no
\end{eqnarray}
\end{thm}
\begin{proof}
The proof is provided in Appendix \ref{apd:optimal}.
\end{proof}

\begin{rmk}
Theorem \ref{thm:optimal} indicates that the observer cannot decide the sampling time $t_{n+1}$ until she takes the $n^{th}$ observation. The conditional expectation on the right hand side of \eqref{eq:time1} is a function of $\pi_{t_{n}}$, which can only be obtained after making the $n^{th}$ observation. Hence, the optimal sampling time is characterized by the sampling interval, which is the time that the observer should wait after she makes the $n^{th}$ observation, on the left hand side of \eqref{eq:time1}.
\end{rmk}

\begin{rmk}
Using Theorem \ref{thm:optimal}, we now give a heuristic explanation of the operator $\mathcal{G}$ and the iterative function $\eqref{eq:iteration}$. In particular, $V_{n}(\pi)$ is the minimum cost when there are only $n$ sampling rights left. We could choose either to stop, which costs $1-\pi$, or to continue and take another observation at $m$ that minimizes the expectation of the future cost. Therefore, the minimizer $m$ in the definition of the operator $\mathcal{G}$ is the next sampling time, and $\pi_{k}$'s in $\mathcal{G}$ are the posterior probabilities that are consistent with the expressions \eqref{eq:posterior}-\eqref{eq:kernel1}.
\end{rmk}

Let
$$ \bar{\pi} = 1 - \pi, \quad \bar{\rho} = 1-\rho,$$
it is easy to verify that
\begin{eqnarray}
\sum\limits_{k=0}^{m-1}\pi_{k} &=& m-\frac{\bar{\pi}}{\rho}(1-\bar{\rho}^{m}), \label{eq:postk} \\
\pi_{m} &=& \frac{(1-\bar{\pi}\bar{\rho}^{m})f_1(X_{m})} {(1-\bar{\pi}\bar{\rho}^{m})f_1(X_{m})+(\bar{\pi}\bar{\rho}^{m})f_0(X_{m})}. \label{eq:postt}
\end{eqnarray}
Hence $\mathcal{G}V(\pi)$ can be simplified as
\begin{eqnarray}
&&\hspace{-10mm} \mathcal{G}V(\pi) = \min\left\{1-\pi, \right. \no\\
&& \left. \inf_{m \geq 1}\left\{ c\left(m-\frac{\bar{\pi}}{\rho}(1-\bar{\rho}^{m})\right)+ \mathbb{E}_{\pi}\left[V(\pi_{m})\right] \right\} \right\},
\end{eqnarray}
and $V_{0}(\pi)$ can be simplified as
\begin{eqnarray}
V_{0}(\pi) = \min_{m\geq 0} \left[c\left(m-\frac{\bar{\pi}}{\rho}(1-\bar{\rho}^{m})\right) + \bar{\pi}\bar{\rho}^{m}\right].
\end{eqnarray}
Based on this form, the optimal stopping time can be further simplified to a threshold rule. We define
\begin{eqnarray}
\pi^{U}_{n} = \inf \{ \pi \in [0, 1]| 1-\pi = V_{N-n}(\pi)\}, \no
\end{eqnarray}
for $n=0, \ldots, N$, and the threshold rule is described in the following theorem.
\begin{thm}\label{thm:concave}
For each $n\leq N$, $V_n(\pi)$ is a concave function of $\pi$ and $V_n(1)=0$. Furthermore, the optimal stopping rule for the $N$ sampling right problem can be given as a threshold rule. Specifically,
\begin{eqnarray}
\eta^{*}&=&\min\{n:\pi_{t_n^{*}} \in \mathcal{S}_{n}\}, \label{eq:th_eta}
\end{eqnarray}
where
\begin{eqnarray}
\mathcal{S}_{n} = \{ \pi_{t_{n}} : \pi_{t_{n}} \geq \pi_{n}^{U} \}
\end{eqnarray}
for $n = 0, \ldots, N-1$ and $\mathcal{S}_{N}=[0, 1]$. Moreover, if  $\eta^{*} < N$, then $\tau^{*} = t_{\eta^{*}}$; if $\eta^{*} = N$, then
\begin{eqnarray}
\tau^{*} = \inf \left\{ k \geq t_{N}: \pi_{k} \geq \pi^{U}_{N} \right\}. \label{eq:th_tau}
\end{eqnarray}
\end{thm}
\begin{proof}
The proof is provided in Appendix~\ref{apd:concave}.
\end{proof}

\begin{rmk}
We notice that $\eta^{*}$ is a threshold rule if $\eta^{*} < N$, but it is not a threshold rule if $\eta^{*} = N$ in Theorem \ref{thm:concave}. Hence $\eta^{*} = N$ is true even if $\pi_{t_{N}^{*}} < \pi^{U}_{N}$. This is consistent with our intuition that the observer cannot take more than $N$ observations. However, on the event $\{ \pi_{t_{N}^{*}} < \pi^{U}_{N} \}$, the optimal stopping rule is still a threshold rule due to the fact that $V_{0}(\pi)$ is concave and $V_{0}(\pi)$ is bounded by $1-\pi$.
\end{rmk}

Although Theorem \ref{thm:concave} simplifies the optimal stopping rule into a threshold rule, the optimal strategy still has a very complex structure as the optimal sampling rule is in general difficult to characterize. From (15), one can see that the optimal sampling rule depends on $V_n(\pi)$. Generally $V_n(\pi)$ does not have a close form for a general value of $n$, and it could only be calculated numerically. For reader's convenience, Table \ref{tbl:fixed_algorithm} summarizes the numerical procedure for the calculation of the optimal solution. Although the problem can be solved numerically, numerical calculation provides little insight for the optimal solution. This motivates us to conduct asymptotic analysis in the next subsection.

\begin{table*}\centering
\caption{Optimal Algorithm for $N$ sampling right Problem}\label{tbl:fixed_algorithm}
\begin{tabular}{ l l l }
  \hline
  \hline
  Offline Procedure: &  & \\
  &step $0$: & Calculate $V_{0}(\pi) = \min\limits_{m\geq 0}\left[c\sum\limits_{k=0}^{m-1}\pi_k+1-\pi_m\right]$.\\ 
  &          & Calculate $W_{0}(\pi, m) = c\left(m-\frac{\bar{\pi}}{\rho}(1-\bar{\rho}^{m}) \right)+\mathbb{E}_{\pi}[V_{0}(\pi_{m})]$. \\
  &          & Calculate $\pi^{U}_{N} = \inf \{ \pi \in [0, 1]| 1-\pi = V_{0}(\pi)\}$. \\
  &step $n$: & Given $W_{n-1}(\pi, m)$, calculate $V_{n}(\pi)=\min\{1-\pi, \inf_{m}W_{n-1}(\pi, m)\}$. \\
  &          & Given $V_{n}(\pi)$, calculate $W_{n}(\pi, m) = c\left(m-\frac{\bar{\pi}}{\rho}(1-\bar{\rho}^{m}) \right)+\mathbb{E}_{\pi}[V_{n}(\pi_{m})]$. \\
  &          & Calculate $\pi^{U}_{N-n} = \inf \{ \pi \in [0, 1]| 1-\pi = V_{n}(\pi)\}$, \\
  &          & for $n=1, 2, \ldots, N.$\\
  \hline
  Online Procedure: &  & \\
  &step $0$: & If $\pi_{0} \geq \pi_{0}^{U}$, the observer stops. Otherwise, continues. \\
  &          & Find the sampling interval $t_{1} = \arg_{m}W_{N}(\pi_{0}, m)$. \\
  &          & Take observation $X_{t_{1}}$ and calculate $\pi_{t_1}$ by \eqref{eq:kernel1}. \\
  &step $n$: & If $\pi_{t_n} \geq \pi_{n}^{U}$, the observer stops. Otherwise, continues. \\
  &          & Find the sampling interval $t_{n+1}-t_{n} = \arg_{m}W_{N-n}(\pi_{n}, m)$. \\
  &          & Take observation $X_{t_{n+1}}$ and calculate $\pi_{t_{n+1}}$ by \eqref{eq:kernel1}, \\
  &          & for $n=1, 2, \ldots, N-1.$\\
  &step $N$: & If $\pi_{t_N} \geq \pi_{N}^{U}$, the observer stops. Otherwise, continues. \\
  &          & Updates the posterior probability by \eqref{eq:kernel0} at every time slot, stops when $\pi_{N}^{U}$ is exceeded. \\
  \hline
\end{tabular}
\end{table*}

\subsection{Asymptotic Upper Bound}
In this subsection, we investigate if there are any scenarios under which the performance of the limited sampling right problem would approach to the performance of the classic Bayesian detection.

The performance of the classic Bayesian case, in which the observer can take observations at every time slot, is certainly a lower bound of the performance of the $N$ sampling right problem. In this case, the asymptotic performance is given in \cite{Tartakovsky:TPIA:04}. Hence we have
\begin{eqnarray}
&&\hspace{-10mm}\mathrm{ADD}(\pi, N, \tau^{*}, \mu^{*}) \no\\
&&\geq \frac{|\log\alpha|}{D(f_{1} || f_{0})+|\log(1-\rho)|}(1+o(1)), \label{eq:lower}
\end{eqnarray}
where $D(f_{1} || f_{0})$ is the Kullback-Leibler (KL) divergence of $f_{1}$ and $f_{0}$.

We consider a uniform sampling strategy with a threshold stopping rule. In particular, the observer adopts a sampling strategy $\mu=\{\varsigma, 2\varsigma, \ldots, \eta\varsigma\}$, i.e., she takes observations every $\varsigma$ symbols, and she adopts a stopping rule $\tau = \inf \{n\varsigma: \pi_{n\varsigma} \geq 1-\alpha, n \in \mathbb{N}\}$. The performance of this uniform sampling strategy serves as an upper bound of the performance of the $N$ sampling right problem. In particular, we have the following proposition:
\begin{prop} \label{prop:quasi}
(Asymptotic Upper Bound) As $\alpha \rightarrow 0$, if the number of sampling rights satisfies
\begin{eqnarray}
N \geq \frac{|\log \alpha|}{|\log(1-\rho)| \varsigma} \label{eq:N_large}
\end{eqnarray}
for some constant $\varsigma < \infty$, then
\begin{eqnarray}
&&\hspace{-10mm} \mathrm{ADD}(\pi, N, \tau^{*}, \mu^{*}) \no\\
&&\leq \frac{|\log\alpha|\varsigma}{D(f_{1}||f_{0})+|\log(1-\rho)|\varsigma}(1+o(1)). \label{eq:upper}
\end{eqnarray}
\end{prop}
\begin{proof}
The proof is provided in Appendix \ref{apd:quasi}.
\end{proof}

\begin{rmk}
In the conventional asymptotic analysis, one is interested in the average detection delay when $\alpha \rightarrow 0$. For the limited observation case ($ 0 \leq N < \infty$), it is easy to find that
\begin{eqnarray}
\mathrm{ADD}(\pi, N, \tau^{*}, \mu^{*})  = \frac{|\log\alpha|}{|\log(1-\rho)|}(1+o(1)). \label{eq:trivial}
\end{eqnarray}
However, this result brings little information since this ADD can be achieved by any sampling strategy with the threshold rule $\tau=\inf\{k, \pi_{k} \geq 1-\alpha\}$.  \eqref{eq:trivial} could only indicate the order of the average detection delay of the limited sampling right problem. In order to obtain an informative result, in Proposition \ref{prop:quasi}, we consider an alternative condition \eqref{eq:N_large}. This condition is weaker than the limited sampling rights constraint, but is stronger than the condition that the observer has infinity many sampling rights, which is assumed in the classic Bayesian setting.
\end{rmk}

\begin{rmk} \label{rmk:N_large}
One can notice from \eqref{eq:N_large} that $N\rightarrow \infty$ when $\alpha \rightarrow 0$ for any given $\rho$. However, this is different from the classic Bayesian quickest detection. In the classic Bayesian problem, the observer has so many sampling rights that she can take observation at every time slot. But \eqref{eq:N_large} cannot guarantee the observer can achieve the false alarm constraint at her last sampling right if she takes sample at every time instance. It guarantees only that one can achieve the false alarm constraint by the uniform sampling with interval $\varsigma$.
\end{rmk}

From Proposition \ref{prop:quasi}, we can identify scenarios under which the performance of the $N$ sampling right problem is close to that of the classic Bayesian problem. Here we give two such cases. In the first case, when $N$ satisfies \eqref{eq:N_large} with $\varsigma = 1$, from \eqref{eq:lower} and \eqref{eq:upper}, we can see that the upper bound and the lower bound are identical, and hence the ADD of the $N$ sampling right problem will be close to that of the classic Baysian problem. For a problem with a finite sampling rights $N$, this condition can be achieved when $\rho \rightarrow 1$. Intuitively, in the large $\rho$ case, even a few samples can lead to a small false alarm probability, hence the $N$ sampling right problem is close to the classic Bayesian problem. In another scenario, if $D(f_{1}||f_{0})$ close to $0$, i.e. $f_{0}$ and $f_{1}$ are very close to each other, the difference between the ADD of the $N$ sampling right problem and that of the classic Bayesian problem is on the order $o(\log \alpha)$. 
Intuitively, in this scenario, the information provided by the likelihood ratios of observations is quite limited, and therefore, the decision making mainly depends on the prior probability of the change-point $\Lambda$.

\section{Problems with the Stochastic Sampling Right Constraint} \label{sec:stc_opt}
In this section, we study the optimal solution for the problem in the general setup when $\nu$ is a stochastic process described in Section \ref{sec:model}.
\subsection{Optimal Solution}
Denote the posterior probability as
\begin{eqnarray}
\pi_{k} = P_{\pi}^{\nu}(\tau \leq k | \mathcal{F}_{k}). \no
\end{eqnarray}

Following the similar procedure as in Proportion \ref{prop:cost}, for any $\mu$ and $\tau$, we can convert the cost function into following form:
\begin{eqnarray}
U(\pi, N, \tau, \mu) = \mathbb{E}_{\pi}^{\nu}\left[ 1-\pi_{\tau}+c\sum_{k=0}^{\tau-1} \pi_{k} \right].
\end{eqnarray}

This problem can be solved by the backward induction method. In particular, we first solve a finite horizon problem, then we extend the solution to the infinite horizon problem by a limit argument. Hence, we first consider a finite horizon problem with a horizon $T$, that is, we consider the case that the observer must stop at a time no later than $T$. We define
\begin{eqnarray}
J^{T}_{k}(\pi_{k}, N_{k}) &\triangleq& \inf_{\mu^{T}_{k+1}\in\mathcal{U}_{k+1}^{T}, \tau \in \mathcal{T}_{k}^{T}} U(\pi_{k}, N_{k}, \tau, \mu_{k+1}^{T}), \no
\end{eqnarray}
with
\begin{eqnarray}
U(\pi_{k}, N_{k}, \tau, \mu_{k+1}^{T}) &\triangleq& \mathbb{E}_{\pi_{k}}^{\nu}\left[1-\pi_{\tau}+c\sum_{i=k}^{\tau-1}\pi_{i}\right], \no
\end{eqnarray}
in which $\mu^{T}_{k} = \{ \mu_{k}, \mu_{k+1}, \ldots, \mu_{T}\}$ is the strategy adopted by the observer from $k$ to $T$, $\mathcal{U}_{k}^{T} = \{\mu^{T}_{k} : N_{i}\geq 0, \forall i=k, \ldots, T \}$ is the admissible set of sampling strategies, and $\mathcal{T}_{k}^{T} = \{ \tau \in \mathcal{T}: k \leq \tau \leq T\}$ is the set of admissible stopping times. We notice that by setting $k=0$, $J_{0}^{T}(\pi_{0}, N_{0})$ is the cost function for the finite horizon problem with a horizon $T$.

We then introduce a set of iteratively defined functions. Let
\begin{eqnarray}
V_{T}^{T}(\pi_{T}, N_{T}) &=& 1-\pi_{T}, \no
\end{eqnarray}
and for $k=T-1, T-2, \ldots, 0$, we define
\begin{eqnarray}
&&\hspace{-10mm} W_{k+1}^{T}(\pi_{k}, N_{k}, \nu_{k+1}) \no\\
&=& \min\left\{\mathbb{E}_{\pi_{k}}^{\nu}[V_{k+1}^{T}(\pi_{k+1}, N_{k+1})|\nu_{k+1}, \mu_{k+1}=0], \right. \no\\
&& \left.\mathbb{E}_{\pi_{k}}^{\nu}[V_{k+1}^{T}(\pi_{k+1}, N_{k+1})|\nu_{k+1}, \mu_{k+1}=1]\right\},\no\\
&&\hspace{-10mm} V_{k}^{T}(\pi_{k}, N_{k}) \no\\
&=& \min\{1-\pi_{k}, c\pi_{k}+\mathbb{E}^{\nu}[W^{T}_{k+1}(\pi_{k}, N_{k},\nu_{k+1})]\}. \no
\end{eqnarray}

This set of functions convert the finite horizon problem into a Markov stopping problem. Specifically, we have the following theorem:
\begin{thm}\label{thm:DP}
For all $k=1, 2, \ldots, T$, we have
\begin{eqnarray}
J^{T}_{k}(\pi_{k}, N_{k}) = V^{T}_{k}(\pi_{k}, N_{k}). \no
\end{eqnarray}
Furthermore, the optimal sampling strategy is given as
\begin{eqnarray}
\mu^{*}_{k} = \argmin_{\mu_{k}\in \{0, 1\}} \mathbb{E}_{\pi_{k-1}}^{\nu}[V_{k}^{T}(\pi_{k}, N_{k})|\nu_{k}, \mu_{k}].\no
\end{eqnarray}
The optimal stopping rule is given as
\begin{eqnarray}
&&\hspace{-15mm}\tau^{*} = \inf\left\{ 0 \leq k \leq T: 1-\pi_{k} \right. \no\\
&&\left. \leq c\pi_{k}+\mathbb{E}^{\nu}[W^{T}_{k+1}(\pi_{k}, N_{k},\nu_{k+1})] \right\}. \no
\end{eqnarray}
\end{thm}
\begin{proof}
This proof is provided in Appendix \ref{apd:DP}.
\end{proof}

\begin{rmk}
Using Theorem \ref{thm:DP}, we now give a heuristic explanation of the iterative functions $W_{k+1}^{T}$ and $V_{k}^{T}$.
In each time slot, as shown in Figure \ref{fig:sys_model}, the observer needs to make two decisions: the sampling decision $\mu_{k}$ and the terminal decision $\delta_{k}$. Both decisions affect the cost function, however these two decisions are based on different information. In particular, the observer decides whether to take an observation or not at time slot $k$ after she knows how many sampling rights has been collected at time slot $k$. Hence, $\mu_{k}$ is a function of $\nu_{k}$, $\pi_{k-1}$ and $N_{k-1}$. When $\mu_{k}$ is decided, the observer could determine the way that $\pi_{k}$ and $N_{k}$ evolve, and hence the decision $\delta_{k}$ is a function of $\pi_{k}$ and $N_{k}$. Actually, the iterative function $V_{k}^{T}$ is the cost function associated with $\delta_{k}$, and $W_{k}^{T}$ is that associated with $\mu_{k}$. At the end of time slot $k$, the observer could choose either to stop, which costs $1-\pi_{k}$, or to continue. Since $\mu_{k+1}$ is the next decision after $\delta_{k}$, the future cost in $V_{k}^{T}$ is $\mathbb{E}^{\nu}[W_{k+1}^{T}]$. 
On the other hand, since $\delta_{k+1}$ is the decision after $\mu_{k+1}$, hence the observer chooses $\mu_{k+1}$ based on the rule that the future cost is minimized, that is the conditional expectation of $V_{k+1}^{T}$ is minimized, which leads the expression of $W_{k+1}^{T}$.
\end{rmk}

In the following, we use a limit argument to extend the above conclusion to the infinite horizon problem. Since $V_{k}^{T}(\pi_{k}, N_{k}) \geq 0$ and
\begin{eqnarray}
V_{k}^{T+1}(\pi_{k}, N_{k}) \leq V_{k}^{T}(\pi_{k}, N_{k}), \no
\end{eqnarray}
which is true due to the fact that all strategies admissible for horizon $T$ are also admissible for horizon $T+1$. As the result, the limit of $V_{k}^{T}(\pi_{k}, N_{k})$ as $T\rightarrow \infty$ exists. Furthermore, as $\pi_{k}$ and $N_{k}$ are homogenous Markov chains, the form of the limit function is the same for different values of $k$, which we define as
\begin{eqnarray}
V(\pi_{k}, N_{k}) \triangleq \lim_{T \rightarrow \infty} V_{k}^{T}(\pi_{k}, N_{k}). \no
\end{eqnarray}
Similarly, we have
\begin{eqnarray}
W(\pi_{k}, N_{k}, \nu_{k+1}) \triangleq \lim_{T \rightarrow \infty} W_{k+1}^{T}(\pi_{k}, N_{k}, \nu_{k+1}). \no
\end{eqnarray}

By the monotone convergence theorem, the iterative functions can be written as
\begin{eqnarray}
&&\hspace{-10mm} W(\pi_{k}, N_{k}, \nu_{k+1}) \no\\
&=& \min\left\{\mathbb{E}_{\pi_{k}}^{\nu}[V(\pi_{k+1}, N_{k+1})|\nu_{k+1}, \mu_{k+1}=0], \right. \no\\
&& \left.\mathbb{E}_{\pi_{k}}^{\nu}[V(\pi_{k+1}, N_{k+1})|\nu_{k+1}, \mu_{k+1}=1]\right\},\no\\
&&\hspace{-10mm} V(\pi_{k}, N_{k}) \no\\
&=& \min\{1-\pi_{k}, c\pi_{k}+\mathbb{E}^{\nu}[W(\pi_{k}, N_{k},\nu_{k+1})]\}. \no
\end{eqnarray}
Hence, we have the following conclusion for the infinite horizon problem.
\begin{thm} \label{thm:opt_p2}
The optimal sampling strategy for (P2) is given as
\begin{eqnarray}
\mu^{*}_{k} = \argmin_{\mu_{k}\in \{0, 1\}} \mathbb{E}_{\pi_{k-1}}^{\nu}[V(\pi_{k}, N_{k})|\nu_{k}, \mu_{k}].
\end{eqnarray}
The optimal stopping rule is given as
\begin{eqnarray}
\tau^{*} = \inf\left\{ k \geq 0: 1-\pi_{k} \leq c\pi_{k} +\mathbb{E}^{\nu}[W(\pi_{k}, N_{k},\nu_{k+1})] \right\}.
\end{eqnarray}
\end{thm}

\subsection{Asymptotically Optimal Solution}
The optimal solution for the stochastic sampling problem has a very complex structure. In this subsection, we propose a low complexity algorithm and show that it is asymptotically optimal when $\alpha \rightarrow 0$. The proposed algorithm is
\begin{eqnarray}
\tilde{\mu}_{k}^{*}= \left\{ \begin{array}{ll}
1 & \textrm{if  } N_{k-1} + \nu_{k} \geq 1 \\
0 & \textrm{if  } N_{k-1} + \nu_{k} = 0
\end{array} \right. ,
\end{eqnarray}
and
\begin{eqnarray}
\tilde{\tau}^{*} = \inf \{ k \geq 0 | \pi_{k} \geq 1-\alpha \}.
\end{eqnarray}
That is, the observer adopts a greedy sampling strategy in which she takes observations as long as she has sampling rights left, and she declares that the change has occurred when the posterior probability exceeds a pre-designed threshold. In the following, we show the asymptotic optimality of this algorithm in two steps. In the first step, we derive a lower bound on the average detection delay for any sampling strategy and any stopping rule. In the second step, we show that $(\tilde{\mu}^{*}, \tilde{\tau}^{*})$ achieves this lower bound asymptotically, which then implies that $(\tilde{\mu}^{*}, \tilde{\tau}^{*})$ is asymptotically optimal. To proceed, we define the likelihood ratio of the observation sequence $\{Z_{k}\}$ as
\begin{eqnarray}
L(Z_{k}) = \left\{ \begin{array}{cc}
                     \frac{f_{1}(X_{k})}{f_{0}(X_{k})}, & \text{if } \mu_{k}=1 \\
                     1, & \text{if } \mu_{k}=0
                   \end{array} \right. , \label{eq:LR}
\end{eqnarray}
and denote $l(Z_{k}) = \log L(Z_{k})$ as the log likelihood ratio. The lower bound on the detection delay is presented in the following theorem:

\begin{thm} \label{thm:lowerbound}
As $\alpha \rightarrow 0$,
\begin{eqnarray}
&&\hspace{-15mm} \inf_{\mu \in \mathcal{U}, \tau \in \mathcal{T}} \mathrm{ADD}(\pi, N, \tau, \mu) \no\\
&& \geq \frac{|\log \alpha|}{\tilde{p}D(f_{1}||f_{0})+|\log(1-\rho)|}(1+o(1)),
\end{eqnarray}
with $\tilde{p} \triangleq \mathbb{E}^{\nu}[\tilde{\mu}^{*}]$.
\end{thm}
\begin{proof}
This proof is provided in Appendix \ref{apd:lowerbound}.
\end{proof}

To study the asymptotic optimality of $(\tilde{\mu}^{*}, \tilde{\tau}^{*})$, we need to impose some additional assumptions on $f_{1}$ and $f_{0}$. Specifically, for any $\varepsilon > 0$, we define the random variable
\begin{eqnarray}
T_{\varepsilon}^{(\lambda)} \triangleq \sup\left\{ n \geq 1: \Big|\frac{1}{n}\sum_{i=\lambda}^{\lambda+n-1} l(Z_{i}) - \tilde{p}D(f_{1}||f_{0}) \Big| > \varepsilon \right\}, \no
\end{eqnarray}
in which the supremum of an empty set is defined as $0$. Under the sampling strategy $\tilde{\mu}^{*}$, we make additional assumptions that
\begin{eqnarray}
\mathbb{E}_{\lambda}^{\nu}\left[ T_{\varepsilon}^{(\lambda)} \right] < \infty \quad \forall \varepsilon > 0 \text{ and } \forall \lambda \geq 1 \label{eq:quickly_converge}
\end{eqnarray}
and
\begin{eqnarray}
\mathbb{E}_{\pi}^{\nu}\left[ T_{\varepsilon}^{(\Lambda)} \right] = \sum_{\lambda=1}^{\infty} \mathbb{E}_{\lambda}^{\nu}\left[ T_{\varepsilon}^{(\lambda)} \right]P(\Lambda = \lambda) < \infty, \quad \forall \varepsilon > 0. \label{eq:averege_quickly}
\end{eqnarray}
With these assumptions, we have following result:

\begin{thm} \label{thm:asym}
If \eqref{eq:quickly_converge} and \eqref{eq:averege_quickly} hold, then $(\tilde{\mu}^{*}, \tilde{\tau}^{*})$ is asymptotically optimal as $\alpha \rightarrow 0$. Specifically,
\begin{eqnarray}
&&\hspace{-15mm}\mathrm{ADD}(\pi, N, \tilde{\tau}^{*}, \tilde{\mu}^{*}) \no\\
&&= \frac{|\log \alpha|}{\tilde{p}D(f_{1}||f_{0})+|\log(1-\rho)|}(1+o(1)).
\end{eqnarray}
\end{thm}
\begin{proof}
This proof is provided in Appendix \ref{apd:asym}.
\end{proof}

\begin{rmk}
More general assumptions corresponding to \eqref{eq:quickly_converge} and \eqref{eq:averege_quickly} are termed as ``$r$-quick convergence'' and ``average-$r$-quick convergence''\cite{Tartakovsky:TPIA:04}, respectively. In particular, \eqref{eq:quickly_converge} and \eqref{eq:averege_quickly} are special cases for $r=1$. The ``$r$-quick convergence'' was originally introduced in \cite{Lai:AnP:76} and has been used previously in \cite{Tartakovsky:SISP:98, Dragalin:TIT:99} to show the asymptotic optimality of the sequential multi-hypothesis test. The ``average-$r$-quick convergence'' was introduced in \cite{Tartakovsky:TPIA:04} to show asymptotic optimality of the Shiryaev-Roberts (SR) procedure in the Bayesian quickest change-point problem.
\end{rmk}

\begin{rmk}
The above theorems indicate that $N_{0}$ does not affect the asymptotic optimality. Since the detection delay goes
to infinity as $\alpha \rightarrow 0$, a finite initial $N_{0}$, which could contribute only a finite number of observations, does not reduce the average detection delay significantly. However, the sampling right capacity $C$ could affect the average detection delay since $\tilde{p}$ is a function of $C$ and $\nu$.
\end{rmk}

\begin{rmk}
Since there is no penalty on the observation cost before the change-point, one may expect the observer to take observations as early as possible for the quickest detection purpose, and hence expect the greedy sampling strategy to be exactly optimal. However, taking observations too aggressively before the change-point will affect how many sampling rights the observer can use after the change-point, although there is no penalty on the observations cost before the change-point.  Theorem \ref{thm:opt_p2} shows that the optimal sampling strategy should be a function of $\pi_{k}$, $N_{k}$ and $\nu_{k}$. Intuitively, an observer will save the sampling rights for future use when she has little energy left ($N_{k}$ is small) or when she is pretty sure that the change-point has not occurred yet ($\pi_{k}$ is small). To use the greedy sampling at the very beginning may reduce the observer's sampling rights at the time when the change occurs, hence increase the detection delay. Therefore, the greedy sampling strategy is only \emph{first order asymptotically optimal} but not exactly optimal.
\end{rmk}

\begin{rmk}
In our recent work \cite{Jun:TSP:13}, we also show that the greedy sampling strategy is asymptotically optimal for the non-Bayesian quickest change-point detection problem with a stochastic energy constraint. Here, we provide a high-level explanation why the greedy sampling strategy performs well for both Bayesian and non-Bayesian case. In asymptotic analysis of both cases (either PFA goes to zero or the average run length to false alarm goes to infinity), the detection delay goes to infinity, hence the observer needs infinitely many sample rights after the change-point. These sample rights mainly come from the replenishing procedure $\nu_{k}$. After the change-point, the greedy sampling strategy is the most efficient way to consume the sampling rights collected by the observer. Before the change-point, the greedy sampling might not be the best strategy, but the penalty incurred by this sub-optimality in terms of the detection delay is at most $C$ (the finite sampling right capacity of the observer), which is negligible when the detection delay goes to infinity. 
\end{rmk}

\section{Numerical Simulation} \label{sec:simulation}
In this section, we give some numerical examples to illustrate the analytical results of the previous sections. In these numerical examples, we assume that the pre-change distribution $f_{0}$ is Gaussian with mean 0 and variance $\sigma^2$. The post-change distribution $f_{1}$ is Gaussian distribution with mean 0 and variance $P+\sigma^2$. In this case, the KL divergence is $D(f_{1} || f_{0}) = \frac{1}{2}\left[ \log \frac{1}{1+P/\sigma^2} + \frac{P}{\sigma^{2}}\right]$. And we denote $SNR = 10\log(P/\sigma^2)$.

The first set of simulations are related to the limited sampling problem. In the first scenario, we illustrate the relationship between ADD and PFA with respect to $N$. In this simulation, we take $\pi_{0} = 0$, $\rho = 0.1$ and $SNR = 0 \mathrm{dB}$, from which we know that $D(f_{1}||f_{0}) \approx 0.15$ and $|\log (1-\rho)| \approx 0.11$ in this case.
The simulation results are shown in Figure \ref{fig:ADD_vs_Pfa_1}. In this figure, the blue line with squares is the simulation result for $N=30$, the green line with stars and the red line with circles are the results for $N=15$ and $N=8$, respectively. The black dash line is the performance of the classic Bayesian problem, which serves as a lower bound for the performance of our problem. The black dot dash line is the performance of the uniform sampling case with sampling interval $\varsigma = 11$ (One can verify this value by putting $\alpha = 10^{-5}$ and $N=8$ into \eqref{eq:N_large}), which serves as an upper-bound for the performance of our problem. As we can see, these three lines lie between the upper bound and the lower bound. Furthermore, the more sampling rights the observer has, the shorter detection delay the observer can achieve, and the closer the performance is to the lower bound.

\begin{figure}[thb]
\centering
\includegraphics[width=0.45 \textwidth]{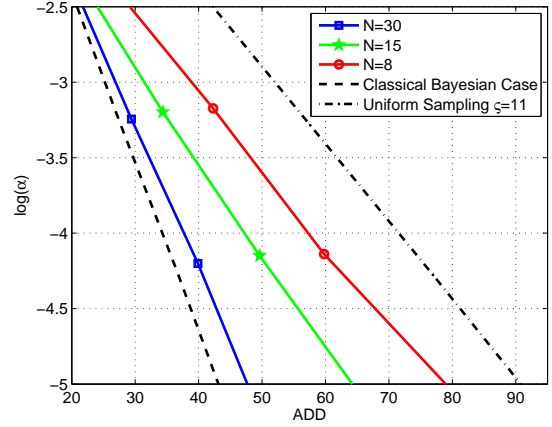}
\caption{PFA v.s. ADD under $SNR = 0\mathrm{dB}$ and $\rho=0.1$ }
\label{fig:ADD_vs_Pfa_1}
\end{figure}

In the second scenario, we discuss the relationship between ADD and PFA with respect to different $\rho$. In this simulation, we set $\pi_{0}=0$, $N=8$ and $SNR = 0\mathrm{dB}$. The simulation results are shown in Figure \ref{fig:ADD_vs_Pfa_3}. In this figure, the red line with circles is the performance with $\rho=0.2$, the green line with stars and the blue line with squares are the performances with $\rho = 0.5$ and $\rho = 0.8$, respectively. The three black dash lines from the top to the bottom are the lower bounds obtained by the classic Bayesian case with $\rho = 0.2$, $\rho = 0.5$ and $\rho = 0.8$, respectively. From this figure we can see that, as $\rho$ increases, the distance between the performance of our scheme and the lower bound is reduced. For the case $\rho = 0.8$, the performance of $N=8$ is almost the same as that of the lower bound, which verifies our analysis that when $\rho$ is large, the performance of limited sampling right problem is close to that of the classic one.

\begin{figure}[thb]
\centering
\includegraphics[width=0.45 \textwidth]{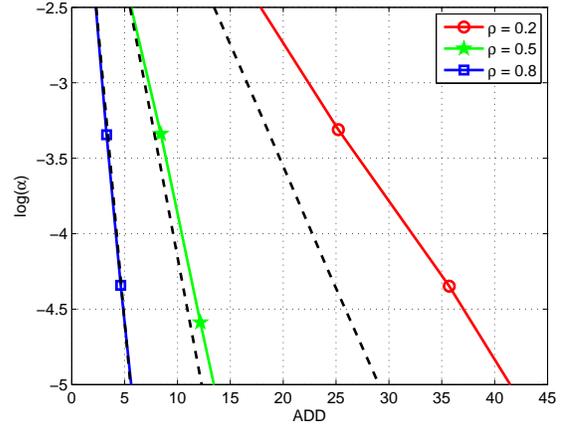}
\caption{PFA v.s. ADD under $SNR = 0\mathrm{dB}$ and $N =8$}
\label{fig:ADD_vs_Pfa_3}
\end{figure}

In the third scenario, we consider the case when $f_{0}$ and $f_{1}$ are close to each other. In the simulation, we set the $SNR = -5\mathrm{dB}$ and $\rho = 0.4$. One can verify that $D(f_{1}|| f_{0}) = 0.02$, which is only about $4\%$ of the value $|\log (1-\rho)|$. In this simulation, we set $N=15$ and $\varsigma =2$ to achieve a false alarm probability $10^{-5}$. The simulation results are shown in Figure \ref{fig:ADD_vs_Pfa_2}. As we can see, the distance between the upper bound, which is the black dot dash line obtained by the uniform sampling with $\varsigma = 2$, and the lower bound, which is the black dash line obtained by the classic Bayesian case, is quite small, and therefore the performance of the limited sampling right problem (the blue line with squares) is quite close to the lower bound.

\begin{figure}[thb]
\centering
\includegraphics[width=0.45 \textwidth]{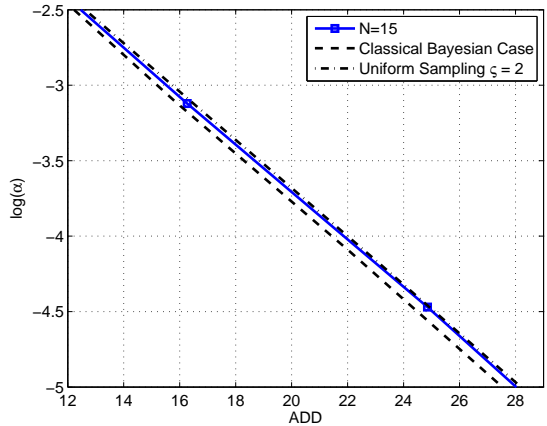}
\caption{PFA v.s. ADD under $SNR = -5\mathrm{dB}$ and $\rho =0.4$}
\label{fig:ADD_vs_Pfa_2}
\end{figure}

In the last simulation, we examine the asymptotic optimality of $(\tilde{\mu}^{*}, \tilde{\tau}^{*})$ for the stochastic sampling right problem. In the simulation, we set $C=3$, and we assume that the amount of sampling right is taken from the set $\mathcal{V} = \{0, 1, \ldots, 4\}$. In this case, the probability transition matrix of the Markov chain $N_{k}$ under $\tilde{\mu}^{*}$ is given as
\begin{eqnarray}
\mathbf{P} = \left[ \begin{array}{c c c c}
p_{0}+p_{1}, &p_{2}, &p_{3}, &p_{4} \\
p_{0}, &p_{1}, &p_{2}, &p_{3}+p_{4} \\
0, &p_{0}, &p_{1}, &\sum_{i=2}^{4} p_{i} \\
0, &0, &p_{0}, &\sum_{i=1}^{4} p_{i} \end{array} \right]. \no
\end{eqnarray}
In the simulation, we set $p_{0}=0.85$, $p_{1}=0.1$, $p_{2}=0.03$, $p_{3}=0.01$, $p_{4}=0.01$, then the stationary distribution is $\tilde{\mathbf{w}}=[0.7988, 0.0988, 0.0624, 0.0390]^{T}$ and $\tilde{p} = 1-p_{0}\tilde{w}_{0} = 0.3610$.
Furthermore, we set $\sigma^{2}=1$ and $SNR = 5\mathrm{dB}$. The simulation result is shown in Figure \ref{fig:general_asym}. In this figure the red line with squares is the performance of the proposed strategy $(\tilde{\tau}^{*}, \tilde{\mu}^{*})$, and the black dash line is calculated by $|\log \alpha|/(\tilde{p}D(f_{1}||f_{0})+|\log(1-\rho)|)$. As we can see, along all the scales, these two curves are parallel to each other, which confirms that the proposed strategy, $(\tilde{\tau}^{*}, \tilde{\mu}^{*})$, is asymptotically optimal as $\alpha \rightarrow 0$ since the constant difference can be ignored when the detection delay goes to infinity.

\begin{figure}[thb]
\centering
\includegraphics[width=0.45 \textwidth]{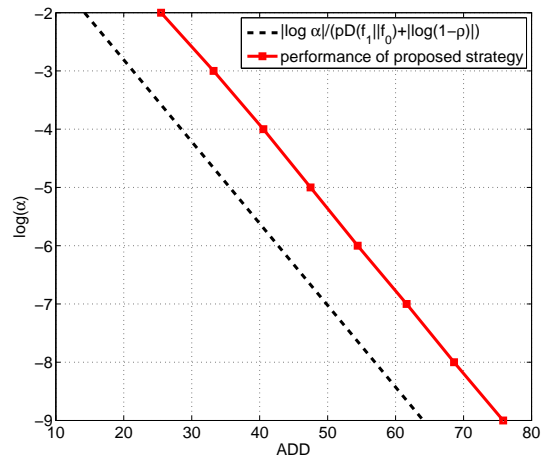}
\caption{PFA v.s. ADD under strategy $(\tilde{\tau}^{*}, \tilde{\mu}^{*})$}
\label{fig:general_asym}
\end{figure}

\section{Conclusion} \label{sec:conclusion}
In this paper, we have analyzed the Bayesian quickest change detection problem with sampling right constraints. Two types of constraints have been considered. The first one is a limited sampling right constraint. We have shown that the cost function of the $N$ sampling right problem can be characterized by a set of iterative functions, each of them could be used for determining the next sampling time or the stopping time. The second constraint is a stochastic sampling right constraint. Under this constraint, we have shown that the greedy sampling strategy coupled with a threshold stopping rule is first order asymptotically optimal as $\alpha \rightarrow 0$.

In terms of future work, it will be interesting to design a low complexity algorithms for the limited sampling right problem. It will also be interesting to develop higher order asymptotically optimal solutions for the stochastic sampling right problem. We will also extend the current work to the distributed sensor network setting.

\appendices

\section{Proof of Lemma \ref{lem:ext}} \label{apd:ext}
Let $\mu=(t_1,\cdots,t_{\eta})$ be a sampling strategy and $\tau = t_{s}$ be a stopping time such $t_{s} > t_{\eta}$ and $\eta < N$. Notice that $t_1,\cdots,t_\eta$ are time instances at which observations are taken, and $t_{s}$ is the time instance at which no sample is taken but the observer announces that a change has occurred. Since $\eta < N$, meaning that there is at least one sampling right left, we construct another strategy $\tilde{\mu}=(t_1,\cdots,t_{\eta}, t_{s})$ and $\tilde{\tau}=t_{s}+m^*$, in which we will take another observation at time $t_{s}$ and then claim that a change has occurred at time $t_{s}+m^{*}$. Here $m^{*}$ is chosen as
\begin{eqnarray}
m^{*} = \argmin_{m \geq 0} H(\pi_{t_{s}},m), \no
\end{eqnarray}
in which
\begin{eqnarray}
H(\pi,m) \triangleq \mathbb{E}_{\pi}\left[c\sum\limits_{k=0}^{m-1}\pi_{k}+1-\pi_m\right] \no
\end{eqnarray}
with
\begin{eqnarray}
\pi_{0} &=& \pi, \no\\
\pi_{k} &=& \pi+\sum\limits_{i=1}^k(1-\pi)\rho(1-\rho)^{i-1}\no\\
&=&\pi+(1-\pi)[1-(1-\rho)^k],\quad k= 1,\ldots m. \no
\end{eqnarray}
Then, we have
\begin{eqnarray}
U(\pi, N, \tilde{\tau}, \tilde{\mu})&=&\mathbb{E}_{\pi}\left[c\sum_{k=0}^{t_{s}+m^{*}-1}\pi_k+1-\pi_{t_{s}+m^{*}}\right]\no\\
                    &=&\mathbb{E}_{\pi}\left[c\sum_{k=0}^{t_{s}-1}\pi_k+H(\pi_{t_{s}},m^*)\right]\no\\
                    &\leq&\mathbb{E}_{\pi}\left[c\sum_{k=0}^{t_{s}-1}\pi_k+H(\pi_{t_{s}},0)\right]\no\\
                    &=&\mathbb{E}_{\pi}\left[c\sum_{k=0}^{t_{s}-1}\pi_k+1-\pi_{t_{s}}\right] \no\\
                    &=& U(\pi, N, \tau, \mu). \no
\end{eqnarray}
Hence, by taking one more observation at time $t_{s}$ and then deciding whether a change has occurred or not can reduce the cost. This implies that if there are sampling rights left, it is not optimal to claim a change without first taking a sample.

\section{Proof of Theorem ~\ref{thm:optimal}} \label{apd:optimal}
We show this theorem by induction: it is clear that $J(\pi, 0)=V_0(\pi)$. Suppose $J(\pi, n-1)=V_{n-1}(\pi)$, we show that $J(\pi, n)=V_n(\pi)$.

Firstly, we show that $J(\pi, n) \geq V_{n}(\pi)$. If the optimal sampling strategy for \eqref{eq:cost} is $t_{\eta} = 0$, then the optimal stopping time is $\tau=0$ by Corollary \ref{cor:ext}. In this case, it is easy to verify that $J(\pi, n) = V_{n}(\pi) = 1-\pi$. Hence the conclusion $J(\pi, n) \geq V_{n}(\pi) $ holds trivially. If the optimal strategy $t_{\eta} \neq 0$, then any given strategy $\mu=\{t_1, \cdots, t_{\eta}\}$ with $t_{1} = 0$ is not optimal, since it simply reduces the set of admissible strategies without bringing any benefit. In the following we consider the sampling strategy with $t_{\eta} \neq 0$ and $t_{1} \neq 0$.

Let $\mu=\{t_1, \cdots, t_{\eta}\}$ be any sampling strategy with $t_{1} \neq 0$ in $\mathcal{U}_{n}$, then we construct another sampling strategy $\tilde{\mu}$ via $\tilde{\mu}=\{t_2,\cdots,t_{\eta}\}$, which is in $\mathcal{U}_{n-1}$. We have
\begin{eqnarray}
&&\hspace{-6mm}U(\pi, n, \tau, \mu)\no\\
	    &=&\mathbb{E}_{\pi}\left[1-\pi_{\tau}+c\sum_{k=0}^{\tau-1}\pi_k\right]\no\\
            &=&\mathbb{E}_{\pi}\left[c\sum_{k=0}^{t_1-1}\pi_k+1-\pi_{\tau}+c\sum_{k=t_1}^{\tau-1}\pi_k\right]\no\\
            &=&\mathbb{E}_{\pi}\left[c\sum_{k=0}^{t_1-1}\pi_k+U(\pi_{t_1}, n-1, \tau, \tilde{\mu})\right]\no\\
            &\geq&\mathbb{E}_{\pi}\left[c\sum_{k=0}^{t_1-1}\pi_k+J(\pi_{t_{1}}, n-1)\right]\no\\
            &\geq&\inf_{m\geq1}\mathbb{E}_{\pi}\left[c\sum_{k=0}^{m-1}\pi_k+ V_{n-1}(\pi_m)\right]\no\\
            &\geq& \min \left\{1-\pi, \inf_{m \geq 1}\mathbb{E}_{\pi}\left[c\sum_{k=0}^{m-1}\pi_k+ V_{n-1}(\pi_m)\right]\right\}. \label{eq:UgeqV}
\end{eqnarray}

Since this is true for any $\mu\in\mathcal{U}_n$ with $t_{1} \neq 0$, and we also know that the strategy $\mu$ with $t_{1} = 0$ could not be optimal unless $t_{\eta} = 0$, then we have
\begin{eqnarray}
J(\pi, n)=\inf\limits_{\mu} U(\pi, n, \tau, \mu)\geq\mathcal{G}V_{n-1}(\pi)=V_n(\pi). \no
\end{eqnarray}

Secondly, we show that $J(\pi, n) \leq V_{n}(\pi)$. Assume the optimal sampling strategy is $\mu^{*}=\{t_{1}^{*}, t_{2}^{*}, \ldots, t_{\eta^{*}}^{*}\} \in \mathcal{U}_{n}$ and the optimal stopping time is $\tau^{*}$, another strategy is denoted as $\mu=\{t_{1}, \tilde{t}_{2},\ldots, \tilde{t}_{\eta}\}$ with stopping time $\tilde{\tau}$, where $t_{1}$ is an arbitrary sampling time, $\tilde{\mu}=\{\tilde{t}_{2},\ldots, \tilde{t}_{n}\}$ with $\tilde{\tau}$ is the optimal strategy achieves $J(\pi_{t_{1}}, n-1) = U(\pi_{t_{1}}, n-1, \tilde{\tau}, \tilde{\mu})$. We have
\begin{eqnarray}
J(\pi, n) &\leq& \mathbb{E}_{\pi} \left[ c\sum_{k=0}^{t_{1}-1}\pi_{k} + J(\pi_{t_{1}}, n-1)\right] \no
\end{eqnarray}
because $(\tilde{\tau}, \mu)$ is not optimal. Since the above inequality holds for every $t_{1}$, we have
\begin{eqnarray}
J(\pi, n) &\leq& \inf_{m \geq 0}\mathbb{E}_{\pi} \left[ c\sum_{k=0}^{m-1}\pi_{k} + V_{n-1}(\pi_{m}) \right] \no\\
&\leq& \inf_{m \geq 1}\mathbb{E}_{\pi} \left[ c\sum_{k=0}^{m-1}\pi_{k} + V_{n-1}(\pi_{m}) \right]. \no
\end{eqnarray}
Moveover, we have
$$J(\pi, n) \overset{(a)}\leq J(\pi, 0) = \inf_{\tau}\mathbb{E}_{\pi}\left[ 1-\pi_{\tau} + c\sum_{k=0}^{\tau-1}\pi_{k} \right]  \overset{(b)} \leq 1-\pi,$$
in which (a) is true because the admissible strategy set of $J(\pi,n)$ is larger than that of $J(\pi,0)$, and (b) is true because $\tau=0$ is not necessarily optimal for $J(\pi, 0)$. Therefore, we have
\begin{eqnarray}
J(\pi, n) &\leq& \min\left\{1-\pi, \inf_{m \geq 1}\mathbb{E}_{\pi} \left[ c\sum_{k=0}^{m-1}\pi_{k} + V_{n-1}(\pi_{m}) \right] \right\} \no \\
 &=& V_{n}(\pi). \no
\end{eqnarray}
Then we can conclude that $J(\pi, n) = V_{n}(\pi)$.

The optimality of ~\eqref{eq:time1} can be verified by putting it into \eqref{eq:UgeqV}, whose inequalities will then become equalities. Further, we can obtain
\begin{eqnarray}
&&\hspace{-15mm} V_{N-n}(\pi_{t_{n}^{*}}) = \min\left\{ 1-\pi_{t_{n}^{*}}, \right. \no\\
&&\left. \mathbb{E}_{\pi_{t_{n}^{*}}}\left[c\sum_{k=0}^{t_{n+1}^{*}-1}\pi_{k} + V_{N-n-1}(\pi_{t_{n+1}^{*}})\right] \right\}. \no
\end{eqnarray}
Notice that $\{\pi_{t_{n}^{*}}\}$ is a Markov chain, hence \eqref{eq:stop} can be immediately obtained by the Markov optimal stopping theorem. By Corollary \ref{cor:ext}, on $\{ \eta^{*} < N\}$ we have $\tau^{*}=t_{\eta^{*}}^{*}$. On $\{ \eta^{*} =N\}$, by \eqref{eq:V0} it is easy to verify that
\begin{eqnarray}
\tau^{*}  - t^{*}_{\eta^{*}} = \argmin_{m \geq 0} \mathbb{E}_{\pi_{t^{*}_{N}}}\left[c\sum\limits_{k=0}^{m-1}\pi_k+1-\pi_m\right]. \no
\end{eqnarray}
Let
$$ m^{*} = \argmin_{m \geq 0} \mathbb{E}_{\pi_{t^{*}_{N}}}\left[c\sum\limits_{k=0}^{m-1}\pi_k+1-\pi_m\right],$$
then
\begin{eqnarray}
\tau^{*} &=& (t^{*}_{\eta^{*}} + m^{*})\mathbf{1}_{\{\eta^{*} = N\}} + t^{*}_{\eta^{*}}\mathbf{1}_{\{\eta^{*} < N\}} \no\\
&=& t^{*}_{\eta^{*}} + m^{*}\mathbf{1}_{\{\eta^{*} = N\}}. \no
\end{eqnarray}

\section{Proof of Theorem~\ref{thm:concave}}\label{apd:concave}
It is easy to see that $0 \leq V_{n}(\pi) \leq 1$ for any $n \leq N$, and $V_{n}(1) = 0$. We next prove the concavity of $V_{n}(\pi)$ by inductive arguments. Clearly $V_0(\pi_k)$ is a concave function of $\pi_k$ and $V_0(1)=0$. Suppose $V_{n-1}(\pi_k)$ is a concave function of $\pi_k$, we show that $V_{n}(\pi_k)$ is a concave function.

We denote
\begin{eqnarray}
A_{n}(\pi) = \mathbb{E}_{\pi}[V_{n-1}(\pi_{m})], \no
\end{eqnarray}
and we show that $A_{n}(\pi)$ is a concave function.

Let $\pi_k^1\in [0,1]$ and $\pi_k^2\in [0,1]$ and $\theta\in [0,1]$, then for any fixed $m$, we have
\begin{eqnarray}
&&\theta A_{n}(\pi_{k}^{1})+(1-\theta) A_{n}(\pi_{k}^{2}) \no\\
&=& \theta \mathbb{E}_{\pi_{k}^{1}}[V_{n-1}(\pi_{k+m}^{1})] + (1-\theta) \mathbb{E}_{\pi_{k}^{2}}[V_{n-1}(\pi_{k+m}^{2})] \no\\
&=& \int (\theta V_{n-1}(\pi^1_{k+m})f(x_{k+m}|\pi_k^1,m) \no \\ &&+(1-\theta)V_{n-1}(\pi^2_{k+m})f(x_{k+m}|\pi_k^2,m)) dx_{k+m}\no\\
&=& \int [\vartheta V_{n-1}(\pi^1_{k+m})+(1-\vartheta) V_{n-1}(\pi^2_{k+m})]\no\\
&&[\theta f(x_{k+m}|\pi_k^1,m)+(1-\theta)f(x_{k+m}|\pi_k^2,m)]dx_{k+m}\no\\
&\overset{(a)}\leq& \int V_{n-1}(\vartheta \pi^1_{k+m}+(1-\vartheta)\pi^2_{k+m}) \no\\
&&[\theta f(x_{k+m}|\pi_k^1,m)+(1-\theta)f(x_{k+m}|\pi_k^2,m)]dx_{k+m} \no
\end{eqnarray}
in which
\begin{eqnarray}
\vartheta=\frac{\theta f(x_{k+m}|\pi_k^1,m)}{\theta f(x_{k+m}|\pi_k^1,m)+(1-\theta)f(x_{k+m}|\pi_k^2,m)}, \no
\end{eqnarray}
and $(a)$ is due to the inductive assumption that $V_{n-1}(\cdot)$ is a concave function.
Now, define
\begin{eqnarray}
\pi_k^3=\theta \pi_k^1+(1-\theta)\pi_k^2, \no
\end{eqnarray}
we can verify that
\begin{eqnarray}
&&\hspace{-6mm}\pi_{k+m}^3 =\no\\
&&\hspace{-6mm} \frac{[1-(1-\pi_k^3)(1-\rho)^{m}]f_1(Y_{k+m})}{[1-(1-\pi_k^3)(1-\rho)^{m}]f_1(Y_{k+m})+(1-\pi_k^3)(1-\rho)^{m}f_0(Y_{k+m})} \no\\
&&\hspace{-6mm}= \vartheta \pi_{k+m}^1+(1-\vartheta)\pi_{k+m}^2. \no
\end{eqnarray}
At the same time, we have
\begin{eqnarray}
\theta f(x_{k+m}|\pi_k^1,m)+(1-\theta)f(x_{k+m}|\pi_k^2,m)
=f(x_{k+m}|\pi_k^3,m). \no
\end{eqnarray}
Hence,
\begin{eqnarray}
\theta A_{n}(\pi_{k}^{1}) + (1-\theta)A_{n}(\pi_{k}^{2}) \leq \mathbb{E}_{\pi_{k}^{3}}\left[V_{n-1}(\pi^3_{k+m})\right] = A_{n}(\pi_{k}^{3}). \no
\end{eqnarray}
Therefore, $A_{n}(\pi) = \mathbb{E}_{\pi}\left[V_{n-1}(\pi_{m})\right] $ is a concave function. As the result, $\inf_{m} \left\{ \mathbb{E}_{\pi}\left[V_{n-1}(\pi_{m})\right]\right\}$ is also concave since it is the minimum of concave function. Then,
\begin{eqnarray}
c\left(m-\frac{\bar{\pi}_k}{\rho}(1-\bar{\rho}^m)\right) + \inf_{m\geq 1} \mathbb{E}_{\pi_{k}}\left[V_{n-1}(\pi_{k+m})\right]
\end{eqnarray}
is also a concave function of $\pi_k$. Further, $V_n(\pi_k)$ is a concave function of $\pi_k$ since it is the minimum of two concave functions.

By the fact that $\{V_{n}(\pi), n=1,\ldots, N\}$ is a family of concave functions, $\{V_{n}(\pi), n=1,\ldots, N\}$ are dominated by $1-\pi$ and $V_{n}(1)=0$, we immediately conclude that $\tau$ is a threshold rule. By Corollary \ref{cor:ext} and Theorem \ref{thm:optimal}, we can easily obtain \eqref{eq:th_eta} and \eqref{eq:th_tau}.

\section{Proof of Proposition \ref{prop:quasi}} \label{apd:quasi}
In the proof, we assume $\pi_{0} = 0$. This assumption will not affect the asymptotic result but will simplify the mathematical derivation.

We consider a uniform sampling scheme with sample interval $\varsigma$. Since it is not optimal for the observer to take an observation every $\varsigma$ time slots, the ADD of the uniform sampling scheme is larger than that of the optimal strategy. Define 
\begin{eqnarray}
\Gamma \triangleq \min \{n: n \varsigma \geq \Lambda\}. \label{eq:quasi}
\end{eqnarray}
The random variable $\Gamma$ acts as the change-point when there is uniform sampling, since from observing $\{ X_{\varsigma}, X_{2 \varsigma}, \ldots \}$, we cannot tell whether the change happens at $\Lambda$ or at $\Gamma \varsigma$. In the following, we derive the ADD when we use $\{X_{k \varsigma}\}$ to detect $\Gamma$, and we use the following stopping rule
\begin{eqnarray}
\gamma = \min\{n: \pi_{n\varsigma} > 1-\alpha \}.
\end{eqnarray}

In the first step, we relax the condition \eqref{eq:N_large} and consider that $N = \infty$. We notice that the problem of detecting $\Gamma$ based on $\{ X_{k \varsigma}\}$ is still under the Bayesian framework. The distribution of $\Gamma$ is given as
\begin{eqnarray}
q_{0} &=& P(\Gamma = 0) = 0, \no \\
q_{k} &=& P(\Gamma = k) = (1-\rho)^{(k-1)\varsigma}\left[ 1 - (1-\rho)^{\varsigma} \right]. \no
\end{eqnarray}
From $(2.6)$ and $(3.1)$ in \cite{Tartakovsky:TPIA:04}, we have
\begin{eqnarray}
d = \lim_{k \rightarrow \infty} \frac{- \log P(\Gamma \geq k+1)}{k} = \varsigma |\log(1-\rho)|. \no
\end{eqnarray}
And on $\{ \Gamma = k \}$
\begin{eqnarray}
\frac{1}{n} \sum_{i=k}^{k+n-1} l(X_{i \varsigma}) \rightarrow D(f_{1} || f_{0}) \quad \text{as} \quad n \rightarrow \infty, \no
\end{eqnarray}
where $l(X_{i \varsigma}) =  \log f_{1}(X_{i \varsigma})/f_{0}(X_{i \varsigma})$ is the log-likelihood ratio. Then, by Theorem 3 in \cite{Tartakovsky:TPIA:04}, we have
\begin{eqnarray}
\mathbb{E}\left[\gamma - \Gamma|\gamma \geq \Gamma \right] \leq \frac{|\log \alpha|}{D(f_{1} || f_{0}) + \varsigma |\log(1-\rho)|} (1+o(1)) \label{eq:add_bound}.
\end{eqnarray}

In the second step, we take \eqref{eq:N_large} into consideration and we show that $P(N \geq \gamma) \rightarrow 1$ as $\alpha \rightarrow 0$. This result indicates that \eqref{eq:N_large} can guarantee that the observer has enough sampling rights so that she can always stop with some sampling rights left. Therefore, \eqref{eq:add_bound} still holds with probability $1$ when we take the constraint \eqref{eq:N_large} into consideration.

By \eqref{eq:N_large}, we have
\begin{eqnarray}
\left( \frac{1}{1-\rho} \right)^{N \varsigma} \geq \frac{1}{\alpha} \quad \text{or} \quad (1-\rho)^{N \varsigma} \leq \alpha \label{eq:N_ineq}.
\end{eqnarray}
Therefore,
\begin{eqnarray}
P(\Gamma \geq N) = \sum_{n=N+1}^{\infty} P(\Gamma = n) = (1-\rho)^{N \varsigma} < \alpha, \no
\end{eqnarray}
and it is clear that $P(\Gamma \geq N) \rightarrow 0$ when $\alpha \rightarrow 0$.

In the following, we show $P(\gamma > N > \Gamma) \rightarrow 0$ as $\alpha \rightarrow 0$. Notice that
\begin{eqnarray}
\{ \gamma > N\} &\Leftrightarrow& \{ \max\{\pi_{0},\ldots,\pi_{N\varsigma}\} < 1-\alpha \} \no \\ &\Leftrightarrow& \cap_{i=0}^{N} \{ \pi_{i\varsigma} < 1-\alpha\}. \no
\end{eqnarray}
Following (3.7) in \cite{Tartakovsky:SADA:10}, we can rewrite $\pi_i$ as
\begin{eqnarray}
\pi_{i\varsigma} = \frac{R_{\rho, i}}{R_{\rho, i}+\frac{1}{1-(1-\rho)^{\varsigma}}}, \label{eq:sr}
\end{eqnarray}
in which
\begin{eqnarray}
R_{\rho, i} = \sum_{k=1}^{i} \prod_{j=k}^{i} \left[ \frac{1}{(1-\rho)^{\varsigma}} L(X_{j \varsigma}) \right], \label{eq:def_r}
\end{eqnarray}
where $L(X_{j \varsigma}) = \frac{f_{1}(X_{j \varsigma})}{f_{0}(X_{j \varsigma})}$ is the likelihood ratio. One can show \eqref{eq:sr} and \eqref{eq:def_r} by inductive argument using \eqref{eq:postt} and $R_{\rho, i}=(1+R_{\rho, i-1})\frac{1}{(1-\rho)^{\varsigma}}L(X_{i\varsigma}).$
Therefore, we have
\begin{eqnarray}
R_{\rho, N} &=& \sum_{k=1}^{N} \prod_{j=k}^{N} \left[ \frac{1}{(1-\rho)^{\varsigma}} L(X_{j \varsigma})\right]   \no\\
&=& \left[ \frac{1}{(1-\rho)^{\varsigma}}\right]^{N} \sum_{k=1}^{N} [(1-\rho)^{\varsigma}]^{k-1}\prod_{j=k}^{N} L(X_{j \varsigma})   \no\\
&\geq& \frac{1}{\alpha} \sum_{k=1}^{N} [(1-\rho)^{\varsigma}]^{k-1}\prod_{j=k}^{N} L(X_{j \varsigma}). \no
\end{eqnarray}
Finally, we have
\begin{eqnarray}
P(\gamma > N > \Gamma) &\leq& P(\gamma > N ) \no\\
&=& P\left( \cap_{i=0}^{N} \{ \pi_{i\varsigma} < 1-\alpha\} \right) \no \\
&\leq& P\left( \pi_{N\varsigma} < 1-\alpha \right) \no\\
&=& P\left( R_{\rho, N} < \frac{1-\alpha}{\alpha} \frac{1}{1-(1-\rho)^{\varsigma}} \right) \no \\
&\leq& P\left( \sum_{k=1}^{N} q_{k} \prod_{j=k}^{N} L(X_{j \varsigma}) < 1-\alpha \right). \label{eq:asym_cond}
\end{eqnarray}
By \eqref{eq:N_large} we have $N \rightarrow \infty$ when $\alpha \rightarrow 0$, hence
\begin{eqnarray}
\sum_{k=1}^{N} q_{k} \prod_{j=k}^{N}L(X_{j \varsigma}) &\rightarrow& \sum_{k=1}^{\infty} q_{k} \prod_{j=k}^{\infty}L(X_{j \varsigma})\no\\
&=&\mathbb{E}_{\pi}\left[ \prod_{k=\Gamma}^{\infty} L(X_{k \varsigma}) \right] = \infty.  \no
\end{eqnarray}
Therefore
\begin{eqnarray}
P(\gamma > N > \Gamma) \leq P(\gamma > N ) \rightarrow 0. \no
\end{eqnarray}
Then
\begin{eqnarray}
P( N \geq \gamma) &=& 1 - P(\Gamma \geq N) - P(\gamma > N > \Gamma) \no\\
&{\rightarrow}& 1.
\end{eqnarray}

As $\alpha \rightarrow 0$, we have
\begin{eqnarray}
\mathbb{E}_{\pi}\left[\gamma - \Gamma | \gamma \geq \Gamma \right] = \frac{\mathbb{E}_{\pi}\left[ (\gamma - \Gamma)^{+} \right]}{1-P(\gamma < \Gamma)} \rightarrow \mathbb{E}_{\pi}\left[ (\gamma - \Gamma)^{+} \right]. \no
\end{eqnarray}
Let $\tau \triangleq \inf\{n\varsigma: \pi_{n\varsigma} > 1-\alpha \} = \gamma \varsigma$. Since $0 \leq \Gamma\varsigma - \Lambda \leq \varsigma-1$ and $\varsigma < \infty$, we obtain
\begin{eqnarray}
&&\mathbb{E}_{\pi}\left[(\tau - \Lambda)^{+}\right] \no\\
&\leq& \frac{|\log\alpha|\varsigma}{D(f_{1}||f_{0})+|\log(1-\rho)|\varsigma} (1+o(1)) + (\varsigma-1). \no \\
&=& \frac{|\log\alpha|\varsigma}{D(f_{1}||f_{0})+|\log(1-\rho)|\varsigma} (1+o(1)).
\end{eqnarray}
Since the uniform sampling scheme and the stopping time $\tau$ are not optimal, the detection delay of the optimal strategy $(\tau^{*}, \mu^{*})$ is less than $\mathbb{E}_{\pi}\left[(\tau - \Lambda)^{+}\right]$. Hence the conclusion of Proposition \ref{prop:quasi} holds.

\section{Proof of Theorem \ref{thm:DP}} \label{apd:DP}
We show this theorem by induction: it is easy to see that $J_{T}^{T}(\pi_{T}, N_{T}) = V_{T}^{T}(\pi_{T}, N_{T})$. Suppose that
$J_{k+1}^{T}(\pi_{k+1}, N_{k+1}) = V_{k+1}^{T}(\pi_{k+1}, N_{k+1})$, we show $J_{k}^{T}(\pi_{k}, N_{k}) = V_{k}^{T}(\pi_{k}, N_{k})$.

We immediately obtain that $J_{k}^{T}(\pi_{k}, N_{k}) \leq V_{k}^{T}(\pi_{k}, N_{k})$ since $J_{k}^{T}(\pi_{k}, N_{k})$ is defined as the minimum cost over $\mathcal{T}_{k}^{T}$ and $\mathcal{U}_{k+1}^{T}$. In the following, we show that $J_{k}^{T}(\pi_{k}, N_{k}) \geq V_{k}^{T}(\pi_{k}, N_{k})$.

By the recursive formulae of $V^{T}_{k}$ and $W^{T}_{k+1}$, we can obtain
\begin{eqnarray}
&&\hspace{-10mm} V^{T}_{k}(\pi_{k}, N_{k}) \no\\
&=& \min\left\{ 1-\pi_{k}, c\pi_{k}+\mathbb{E}^{\nu}[W^{T}_{k+1}(\pi_{k}, N_{k},\nu_{k+1})]\right\} \no\\
&=& \min\left\{ 1-\pi_{k}, c\pi_{k}+ \sum_{j=0}^{\infty}p_{j}W^{T}_{k+1}(\pi_{k}, N_{k}, j)\right\} \no\\
&=& \min\left\{ 1-\pi_{k}, c\pi_{k}+ \right. \no\\
&&\left.  \sum_{j=0}^{\infty}p_{j}\min\left\{\mathbb{E}_{\pi_{k}}^{\nu}[V_{k+1}^{T}(\pi_{k+1}, N_{k+1})|\nu_{k+1}=j, \mu_{k+1}=0], \right.\right.\no\\
&& \hspace{8mm}\left. \left.\mathbb{E}_{\pi_{k}}^{\nu}[V_{k+1}^{T}(\pi_{k+1}, N_{k+1})|\nu_{k+1}=j, \mu_{k+1}=1]\right\} \right\}. \label{eq:adp_dp1}
\end{eqnarray}

On the other hand, for $J_{k}^{T}(\pi_{k}, N_{k})$ we have
\begin{eqnarray}
&&\hspace{-10mm} J^{T}_{k}(\pi_{k}, N_{k}) \no\\
&=& \inf_{\mu^{T}_{k+1}\in\mathcal{U}_{k+1}^{T}, \tau \in \mathcal{T}_{k}^{T}} \mathbb{E}_{\pi_{k}}^{\nu}\left[1-\pi_{\tau}+c\sum_{i=k}^{\tau-1}\pi_{i}\right] \no \\
&=& \inf_{\mu^{T}_{k+1}\in\mathcal{U}_{k+1}^{T}, \tau \in \mathcal{T}_{k}^{T}} \left[ \mathbb{E}_{\pi_{k}}^{\nu}\left[1-\pi_{\tau}+c\sum_{i=k}^{\tau-1}\pi_{i}\right]\mathbf{1}_{\{\tau=k\}}+ \right. \no\\
&& \left. \mathbb{E}_{\pi_{k}}^{\nu}\left[1-\pi_{\tau}+c\sum_{i=k}^{\tau-1}\pi_{i}\right]\mathbf{1}_{\{\tau\geq k+1\}}\right] \no\\
&=& \inf_{\mu^{T}_{k+1}\in\mathcal{U}_{k+1}^{T}, \tau \in \mathcal{T}_{k}^{T}} \left[ \left(1-\pi_{k}\right)\mathbf{1}_{\{\tau=k\}} + \right. \no\\
&&\left.\mathbb{E}_{\pi_{k}}^{\nu}\left[1-\pi_{\tau}+c\pi_{k} +c\sum_{i=k+1}^{\tau-1}\pi_{i}\right]\mathbf{1}_{\{\tau \geq k+1\}}\right] \no\\
&=& \min\left\{1-\pi_{k}, c\pi_{k}+\right. \no\\
&&\left. \inf_{\mu^{T}_{k+1}\in\mathcal{U}_{k+1}^{T}, \tau \in\mathcal{T}_{k+1}^{T}} \mathbb{E}_{\pi_{k}}^{\nu}\left[1-\pi_{T}+c\sum_{i=k+1}^{T-1}\pi_{i}\right]\right\} \no \\
&=& \min\left\{1-\pi_{k}, c\pi_{k}+\right. \no\\
&&\left. \inf_{\mu^{T}_{k+1}\in\mathcal{U}_{k+1}^{T}, \tau \in\mathcal{T}_{k+1}^{T}} \hspace{-1mm} \mathbb{E}_{\pi_{k}}^{\nu}\hspace{-1mm}\left[\mathbb{E}_{\pi_{k+1}}^{\nu}\hspace{-1mm}\left[1-\pi_{T}+c\sum_{i=k+1}^{T-1}\pi_{i}\hspace{-0.5mm}\right]\hspace{-0.5mm}\right]\hspace{-0.5mm}\right\} \no\\
&=& \min\left\{1-\pi_{k}, c\pi_{k}+\right. \no\\
&&\left. \inf_{\mu^{T}_{k+1}\in\mathcal{U}_{k+1}^{T}, \tau\in\mathcal{T}_{k+1}^{T}} \hspace{-1mm}\mathbb{E}_{\pi_{k}}^{\nu}\hspace{-1mm}\left[U(\pi_{k+1}, N_{k+1}, \tau, \mu_{k+2}^{T})\right]\hspace{-1mm}\right\}. \label{eq:apd_dp2}
\end{eqnarray}
At the same time, we have
\begin{eqnarray}
&&\hspace{-10mm} \mathbb{E}_{\pi_{k}}^{\nu}\left[U(\pi_{k+1}, N_{k+1}, \tau, \mu_{k+2}^{T})\right] \no\\
&=& \sum_{j=0}^{\infty} p_{j}\mathbb{E}_{\pi_{k}}^{\nu}\left[ U(\pi_{k+1}, N_{k+1}, \tau, \mu_{k+2}^{T}) \Bigg|\nu_{k+1}=j\right] \no\\
&\overset{(a)}\geq& \sum_{j=0}^{\infty}p_{j}\min\left\{ \right. \no\\
&&\left. \mathbb{E}_{\pi_{k}}^{\nu}\left[U(\pi_{k+1}, N_{k+1}, \tau, \mu_{k+2}^{T}) \Bigg|\nu_{k+1}=j, \mu_{k+1}=0 \right], \right.\no\\
&&\left. \mathbb{E}_{\pi_{k}}^{\nu}\left[U(\pi_{k+1}, N_{k+1}, \tau, \mu_{k+2}^{T}) \Bigg|\nu_{k+1}=j, \mu_{k+1}=1 \right] \hspace{-1mm} \right\},\no\\\label{eq:apd_dp3}
\end{eqnarray}
in which (a) holds because $\mathbb{E}_{\pi_{k}}^{\nu}\left[ U(\pi_{k+1}, N_{k+1}, \tau, \mu_{k+2}^{T}) |\nu_{k+1}=j\right]$ is a linear combination of \\$\mathbb{E}_{\pi_{k}}^{\nu}\left[U(\pi_{k+1}, N_{k+1}, \tau, \mu_{k+2}^{T}) |\nu_{k+1}=j, \mu_{k+1}=i \right]$ for $i=0, 1$. Substituting \eqref{eq:apd_dp3} into \eqref{eq:apd_dp2}, and using inequalities $\inf (a+b) \geq \inf a + \inf b$, $\inf \min\{a, b\} \geq \min\{ \inf a, \inf b\}$ and $\inf \mathbb{E}[\cdot] \geq \mathbb{E}[\inf(\cdot)]$, we obtain
\begin{eqnarray}
&&\hspace{-10mm} J^{T}_{k}(\pi_{k}, N_{k}) \no\\
&\geq& \min\left\{1-\pi_{k}, c\pi_{k} + \sum_{j=0}^{\infty} p_{j} \min \left\{   \right. \right.\no\\
&& \left. \mathbb{E}_{\pi_{k}}^{\nu}\left[\inf_{\mu^{T}_{k+1}\in\mathcal{U}_{k+1}^{T}, T\in\mathcal{T}_{k+1}^{T}} U(\pi_{k+1}, N_{k+1}, \tau, \mu_{k+2}^{T}) \Bigg|\nu_{k+1}=j, \right.\right. \no\\
&&\left. \mu_{k+1}=0 \right], \no\\
&& \left. \mathbb{E}_{\pi_{k}}^{\nu}\left[ \inf_{\mu^{T}_{k+1}\in\mathcal{U}_{k+1}^{T}, T\in\mathcal{T}_{k+1}^{T}} U(\pi_{k+1}, N_{k+1}, \tau, \mu_{k+2}^{T}) \Bigg|\nu_{k+1}=j, \right. \right. \no\\
&&\left. \left. \mu_{k+1}=1 \right] \right\} \no \\
&=& \sum_{j=0}^{\infty}p_{j} \min\left\{ \right. \no\\
&&\left. \mathbb{E}_{\pi_{k}}^{\nu}\left[J_{k+1}^{T}(\pi_{k+1}, N_{k+1}) \Bigg|\nu_{k+1}=j, \mu_{k+1}=0 \right], \right.\no\\
&&\left. \mathbb{E}_{\pi_{k}}^{\nu}\left[J_{k+1}^{T}(\pi_{k+1}, N_{k+1}) \Bigg|\nu_{k+1}=j, \mu_{k+1}=1 \right] \right\}. \label{eq:adp_dp4}
\end{eqnarray}
Since we assume that $J_{k+1}^{T}(\pi_{k+1}, N_{k+1}) = V_{k+1}^{T}(\pi_{k+1}, N_{k+1})$, by \eqref{eq:adp_dp1} and \eqref{eq:adp_dp4} we can obtain $J_{k}^{T}(\pi_{k}, N_{k}) \geq V_{k}^{T}(\pi_{k}, N_{k})$. 

\section{Proof of Theorem \ref{thm:lowerbound}} \label{apd:lowerbound}
In this proof, we can consider the case that $N_{0}=C$, i.e., the observer has a maximum amount of sampling rights at the beginning. The lower bound for the ADD of this case will certainly be the lower bound for the ADD of the case with $N_{0} < C$. The proof of Theorem \ref{thm:lowerbound} requires several supporting propositions and Theorem 1 in \cite{Tartakovsky:TPIA:04}, which are presented as follows.

\begin{prop}\label{prop:1}
$\mathbb{E}^{\nu}[\tilde{\mu}^{*}]$ exists, and $0 < \mathbb{E}^{\nu}[\tilde{\mu}^{*}] \leq 1$.
\end{prop}
\begin{proof}
The outline of this proof is described as follows: by \eqref{eq:N_evolve}, one can show that $N_{k}$ is a regular Markov chain under $\tilde{\mu}^{*}$. Denote the stationary distribution of $N_{k}$ as $\tilde{\mathbf{w}}=[\tilde{w}_{0}, \tilde{w}_{1}, \ldots, \tilde{w}_{C}]^{T}$, where $\tilde{w}_{i}$ is the stationary probability for the state $N_{k}=i$. By the definition of $\tilde{\mu}^{*}$, it is easy to verify that $\mathbb{E}^{\nu}[\tilde{\mu}^{*}_{k}] = 1-p_{0}\tilde{w}_{0}$ as $k \rightarrow \infty$. Hence the statement holds. The detailed proof of this proposition follows that of Lemma 5.1 in \cite{Jun:TSP:13}, hence we omit the proof here for brevity.
\end{proof}

\begin{prop}\label{prop:convergence}
Given $\Lambda=\lambda$, we have
\begin{eqnarray}
\lim_{r \rightarrow \infty} P_{\lambda}^{\nu} \left\{ \frac{1}{r} \max_{0< h \leq r} \sum_{i=\lambda}^{\lambda+h}l(Z_{i}) \geq (1+\varepsilon)\tilde{p}D(f_{1}||f_{0})\right\} \rightarrow 0  \no\\
\quad \forall \varepsilon > 0,
\end{eqnarray}
where $\tilde{p} = \mathbb{E}[\tilde{\mu}^{*}]$.
\end{prop}
\begin{proof}
Following the proof of Proposition C.1 in \cite{Jun:TSP:13}, we can obtain that the inequality
\begin{eqnarray}
\frac{1}{r} \sum_{i=\lambda}^{r+\lambda-1} l(Z_{i}) \leq \tilde{p}D(f_{1}||f_{0}), \text{ as } r \rightarrow \infty, \label{eq:asure}
\end{eqnarray}
holds almost surely under $P_{\lambda}^{\nu}$ for any $\lambda \geq 1$.

For any $\varepsilon > 0$, define
$$\hat{T}_{\varepsilon}^{(\lambda)} = \sup\left\{r\geq 1 \Bigg| \frac{1}{r}\sum_{i=\lambda}^{\lambda+r-1} l(Z_{i}) > (1 + \varepsilon)\tilde{p}D(f_{1}||f_{0}) \right\}.$$
Due to \eqref{eq:asure}, we have
$$P_{\lambda}^{\nu}\left\{\hat{T}_{\varepsilon}^{(\lambda)} < \infty \right\} = 1,$$
which indicates
\begin{eqnarray}
&&\lim_{r \rightarrow \infty} P_{\lambda}^{\nu} \left\{ \frac{1}{r} \max_{0 < h \leq r}  \sum_{i=k}^{k+h} l(Z_{i}) \geq (1+\varepsilon) \tilde{p}D(f_{1}||f_{0}) \right\} \rightarrow 0 \no.
\end{eqnarray}
\end{proof}

Let $q = \tilde{p}D(f_{1}||f_{0})$. From (2.6) in \cite{Tartakovsky:TPIA:04} we have
\begin{eqnarray}
d = - \lim_{k \rightarrow \infty} \frac{\log P(\Lambda \geq k+1)}{k} = |\log (1-\rho)|.
\end{eqnarray}

To prove Theorem \ref{thm:lowerbound}, we need Theorem 1 in \cite{Tartakovsky:TPIA:04} , which is restated as follows:
\begin{lem} \label{thm:Tartakovsky}
\emph{(\cite{Tartakovsky:TPIA:04}, Theorem 1)} Let $\{ Z_{k} \}$ be a sequence of random variables with a random change-point $\Lambda$. Under $\{\Lambda = \lambda\}$, the conditional distribution of $Z_{k}$ is $f_{0}(\cdot| \mathbf{Z}_{1}^{k-1})$ for $k<\lambda$ and is $f_{1}(\cdot| \mathbf{Z}_{1}^{k-1})$ for $k \geq \lambda$. Denote $P_{\infty}$ as the probability measure under $\{\Lambda = \infty\}$. Denote $l(Z_{k})$ as
$$ l(Z_{k}) = \log \frac{f_{1}(Z_{k}|\mathbf{Z}_{1}^{k-1})}{f_{0}(Z_{k}|\mathbf{Z}_{1}^{k-1})}.$$
Let
\begin{eqnarray}
d = -  \lim_{k \rightarrow \infty} \frac{\log P(\Lambda \geq k+1)}{k}. \no
\end{eqnarray}
If the condition
\begin{eqnarray}
&& \lim_{r \rightarrow \infty} P_{\lambda} \left\{ \frac{1}{r} \max_{0< h \leq r} \sum_{i=\lambda}^{m+h} l(Z_{i}) \geq (1+\varepsilon)q \right\} \rightarrow 0, \no\\
&& \hspace{40mm} \forall \varepsilon > 0 \text{ and } \forall \lambda \geq 1 \label{eq:necessary_condition}
\end{eqnarray}
holds for some constant $q>0$. Denote $q_{d} = q+d$. Then, for all $r > 0$ as $\alpha \rightarrow 0$,
\begin{eqnarray}
&& \inf_{\tau} \mathbb{E}_{\lambda}[(\tau-\lambda)^{r}|\tau \geq \lambda] \geq \left(\frac{|\log \alpha|}{q_{d}}\right)^{r}(1+o(1)). \no\\
&& \inf_{\tau} \mathbb{E}_{\pi}[(\tau-\Lambda)^{r}|\tau \geq \Lambda] \geq \left(\frac{|\log \alpha|}{q_{d}}\right)^{r}(1+o(1)). \no
\end{eqnarray}
\end{lem}
\begin{proof}
Please refer to \cite{Tartakovsky:TPIA:04}.
\end{proof}

In our case, for any arbitrary but given sampling strategy $\mu$, the conditional density
\begin{eqnarray}
f_{0}( Z_{k} | \mathbf{Z}_{1}^{k-1} )
&=& f_{0}( X_{k} )P\left(\left\{ \mu_{k}=1 \right\} \right) + \delta( \phi )P\left(\left\{ \mu_{k}=0 \right\} \right),  \no \\
f_{1}( Z_{k} | \mathbf{Z}_{1}^{k-1} )
&=& f_{1}( X_{k} )P\left(\left\{ \mu_{k}=1 \right\} \right) + \delta( \phi )P\left(\left\{ \mu_{k}=0 \right\} \right),  \no
\end{eqnarray}
where $\delta( \phi )$ is the Dirac delta function. Therefore, the log likelihood ratio in Theorem \ref{thm:Tartakovsky} is
$$ l(Z_{k}) =  \log \frac{f_{1}(Z_{k}|\mathbf{Z}_{1}^{k-1})}{f_{0}(Z_{k}|\mathbf{Z}_{1}^{k-1})}
= \left\{ \begin{array}{c c} \log \frac{f_{1}(Z_{k})}{f_{0}(Z_{k})}, &  \text{if } \mu_{k} = 1 \\
0, & \text{if } \mu_{k} = 0  \end{array} , \right. $$
which is consistent with the definition in \eqref{eq:LR}. Moreover, for any sampling strategy, \eqref{eq:necessary_condition} holds for the constant $q=\tilde{p}D(f_{1}||f_{0})$. Correspondingly, $q_{d} = \tilde{p}D(f_{1}||f_{0})+ |\log(1-\rho)|$. Therefore, by choosing $r=1$, and combining Lemma \ref{thm:Tartakovsky} with Propositions~\ref{prop:1} and~\ref{prop:convergence}, we have:
\begin{eqnarray}
&& \inf_{\mu \in \mathcal{U}, \tau \in \mathcal{T}} \mathbb{E}_{\pi}^{\nu} [\tau - \Lambda|\tau \geq \Lambda] \no\\
&\geq& \frac{|\log \alpha|}{\tilde{p}D(f_{1}||f_{0})+|\log(1-\rho)|}(1+o(1)). \no
\end{eqnarray}
Since
\begin{eqnarray}
\mathbb{E}_{\pi}^{\nu} [\tau - \Lambda|\tau \geq \Lambda] = \frac{\mathbb{E}_{\pi}^{\nu} [(\tau - \Lambda)^{+}]}{1 - P_{\pi}^{\nu} (\tau < \Lambda)} \leq \frac{\mathbb{E}_{\pi}^{\nu} [(\tau - \Lambda)^{+}]}{1 - \alpha}, \no
\end{eqnarray}
as $\alpha \rightarrow 0$, we have
\begin{eqnarray}
\inf_{\mu \in \mathcal{U}, \tau \in \mathcal{T}} \mathbb{E}_{\pi}^{\nu} [(\tau - \Lambda)^{+}] \geq \frac{|\log \alpha|}{\tilde{p}D(f_{1}||f_{0})+|\log(1-\rho)|}(1+o(1)). \no
\end{eqnarray}

\section{Proof of Theorem \ref{thm:asym}} \label{apd:asym}
In this appendix we prove that the proposed strategy $(\tilde{\tau}^{*}, \tilde{\mu}^{*})$ can achieve the lower bound presented in Theorem \ref{thm:lowerbound}. In this proof, we can consider the case that $N_{0}=0$, i.e., the observer does not have any sampling rights at the beginning. If the lower bound of the ADD can be achieved by this case, then it must be achievable for the case with $N_{0}>0$. Define
$$R_{k} \triangleq \log \frac{\pi_{k}}{1-\pi_{k}}.$$
The proposed stopping rule can be expressed in terms of $R_{k}$ as
$$\tilde{\tau}^{*} = \inf\left\{k \geq 0: R_{k} \geq \log \frac{1-\alpha}{\alpha} \right\}.$$
Let $b \triangleq \log \frac{1-\alpha}{\alpha}$. As $\alpha \rightarrow 0$, we have $b = |\log \alpha|(1+o(1))$.

By \eqref{eq:recursion}, \eqref{eq:kernel0}, \eqref{eq:kernel1} and \eqref{eq:LR}, it is easy to verify that
\begin{eqnarray}
R_{k} &=& R_{k-1} + l(Z_{k}) + |\log(1-\rho)| + \log\left(1+\rho\frac{1-\pi_{k-1}}{\pi_{k-1}}\right). \no
\end{eqnarray}
Using this recursive formula repeatedly, we obtain
\begin{eqnarray}
R_{k} &=& \sum_{i=1}^{k} l(Z_{i}) + k|\log(1-\rho)| + \log\left( \frac{\pi_{0}}{1-\pi_{0}}+\rho \right) + \no\\
&&\sum_{i=2}^{k}  \log \left(1+\rho\frac{1-\pi_{i-1}}{\pi_{i-1}}\right). \no
\end{eqnarray}
We notice that the third item in the above expression is a constant. Since the threshold $b$ in the proposed stopping rule will go to infinity as $\alpha \rightarrow 0$, this constant item can be ignored in the asymptotic analysis. For simplicity, we assume $\log(\frac{\pi_{0}}{1-\pi_{0}}+\rho) = 0$ in the rest of this appendix.

Let
\begin{eqnarray}
&& S_{k} \triangleq \sum_{i=1}^{k} l(Z_{i}) + k|\log(1-\rho)|, \no\\
&& \tau_{s} \triangleq  \inf\{k \geq 0: S_{k} \geq b \}. \no
\end{eqnarray}
It is easy to see $\tilde{\tau}^{*} \leq \tau_{s}$ since $R_{k} \geq S_{k}$. The following proposition indicates that $\tau_{s}$ can achieve the lower bound presented in Theorem \ref{thm:lowerbound}, hence $\tilde{\tau}^{*}$ is asymptotically optimal.
\begin{prop}
As $b \rightarrow \infty$,
\begin{eqnarray}
&&\hspace{-20mm}\mathbb{E}_{\pi}^{\nu}[\tau_{s} - \Lambda |\tau_{s} \geq \Lambda] \no\\
&\leq& \frac{b}{\tilde{p}D(f_{1}||f_{0})+|\log(1-\rho)|}(1+o(1)).
\end{eqnarray}
\end{prop}
\begin{proof}
On the event $\{ \Lambda = \lambda\}$, we can decompose $S_{n}$ into two parts if $n \geq \lambda$:
\begin{eqnarray}
S_{n} = S_{1}^{\lambda-1} + S_{\lambda}^{n}, \label{eq:twoS}
\end{eqnarray}
where
\begin{eqnarray}
S_{1}^{\lambda-1} \triangleq \sum_{i=1}^{\lambda-1} l(Z_{i}) + (\lambda-1)|\log(1-\rho)|, \no\\
S_{\lambda}^{n} \triangleq \sum_{i=\lambda}^{n} l(Z_{i}) + (n-\lambda+1)|\log(1-\rho)|. \no
\end{eqnarray}

We first show that as $r \rightarrow \infty$
\begin{eqnarray}
\frac{1}{r}S_{\lambda}^{\lambda+r-1} \overset{a.s.}\rightarrow \tilde{p}D(f_{1}||f_{0})+|\log(1-\rho)|. \label{eq:as_converge}
\end{eqnarray}
Let $\hat{r}$ be the number of non-zero elements in $\{ \mu_{\lambda}, \mu_{\lambda+1}, \ldots, \mu_{\lambda+r-1} \}$, then as $r \rightarrow \infty$, we have
\begin{eqnarray}
\frac{\hat{r}}{r} = \frac{1}{r} \sum_{i=\lambda}^{\lambda+r-1} \mu_{i} \overset{a.s.}{\rightarrow} \mathbb{E}[\mu] = \tilde{p}. \no
\end{eqnarray}
Let $\{a_{1}, \ldots, a_{\hat{r}}\}$ be a sequence of time slots in which the observer takes observations after $\lambda$. That is, $\lambda \leq a_{1} < \ldots < a_{\hat{r}} \leq \lambda+r-1$ and $\mu_{a_{i}}=1$. By the strong law of large numbers, as $\hat{r} \rightarrow \infty$
\begin{eqnarray}
\frac{1}{\hat{r}} \sum_{i=1}^{\hat{r}} l(X_{a_{i}}) \overset{a.s.}{\rightarrow} D(f_{1}||f_{0}). \no
\end{eqnarray}
Then we have
\begin{eqnarray}
\frac{1}{r}S_{\lambda}^{\lambda+r-1} &=& \frac{1}{r}\left[\sum_{i=\lambda}^{\lambda+r-1}l(Z_{i}) + r|\log(1-\rho)|\right] \no\\
&=& \frac{\hat{r}}{r} \frac{1}{\hat{r}} \sum_{i=1}^{\hat{r}} l(X_{a_{i}}) + |\log(1-\rho)| \no\\
&\overset{a.s.}\rightarrow& \tilde{p}D(f_{1}||f_{0})+|\log(1-\rho)|. \no
\end{eqnarray}
In the following, we denote $q_{d} = \tilde{p}D(f_{1}||f_{0})+|\log(1-\rho)|$.

By \eqref{eq:twoS}, we can rewrite $\tau_{s}$ as
\begin{eqnarray}
\tau_{s} = \inf\left\{j>0: S_{\lambda}^{j} \geq b - S_{1}^{\lambda-1}\right\}. \no
\end{eqnarray}
Hence,
\begin{eqnarray}
S_{\lambda}^{\tau_{s}-1} < b- S_{1}^{\lambda-1}.
\end{eqnarray}
Define the random variable
\begin{eqnarray}
\tilde{T}_{\varepsilon}^{(\lambda)} \triangleq \sup\left\{ n \geq 1: |n^{-1}S_{\lambda}^{\lambda+n} - q_{d}| > \varepsilon \right\}.\no
\end{eqnarray}
By \eqref{eq:as_converge}, we have $\tilde{T}_{\varepsilon}^{(\lambda)} < \infty$ almost surely. By \eqref{eq:quickly_converge} and \eqref{eq:averege_quickly}, it is easy to verify that $\mathbb{E}_{\lambda}^{\nu}[\tilde{T}_{\varepsilon}^{(\lambda)}] < \infty$ and $\mathbb{E}_{\pi}^{\nu}[\tilde{T}_{\varepsilon}^{(\Lambda)}] < \infty$.

On the event $\left\{ \tau_{s} > \tilde{T}_{\varepsilon}^{(\lambda)} + (\lambda-1) \right\}$, we have
\begin{eqnarray}
S_{\lambda}^{\tau_{s}-1} > (\tau_{s}-\lambda+1)(q_{d}-\varepsilon), \no
\end{eqnarray}
hence
\begin{eqnarray}
\tau_{s}-\lambda+1 < \frac{S_{\lambda}^{\tau_{s}-1}}{q_{d}-\varepsilon} < \frac{b-S_{1}^{\lambda-1}}{q_{d}-\varepsilon}.
\end{eqnarray}

Then we have
\begin{eqnarray}
&&\tau_{s}-\lambda+1 \no\\
&<& \frac{b-S_{1}^{\lambda-1}}{q_{d}-\varepsilon} \mathbf{1}_{\left\{ \tau_{s} > \tilde{T}_{\varepsilon}^{(\lambda)} + (\lambda-1) \right\}} + \tilde{T}_{\varepsilon}^{(\lambda)} \mathbf{1}_{ \left\{\tau_{s} \leq \tilde{T}_{\varepsilon}^{(\lambda)} + (\lambda-1) \right\}} \no\\
&<& \frac{b-S_{1}^{\lambda-1}}{q_{d}-\varepsilon}+ \tilde{T}_{\varepsilon}^{(\lambda)}. \no
\end{eqnarray}
Taking the conditional expectation on both sides, since $\tilde{T}_{\varepsilon}^{(\lambda)} < \infty$, then as $\alpha \rightarrow 0$ ($b \rightarrow \infty$) we have
\begin{eqnarray}
&&\mathbb{E}_{\lambda}^{\nu}[\tau_{s}-\lambda|\tau_{s}\geq \lambda] \no\\
&\leq& \frac{b}{q_{d}-\varepsilon} - \frac{\mathbb{E}_{\lambda}^{\nu}[S_{1}^{\lambda-1}|\tau_{s}\geq \lambda]}{q_{d}-\varepsilon} + \mathbb{E}^{\nu}_{\lambda}[\tilde{T}_{\varepsilon}^{(\lambda)}|\tau_{s}\geq \lambda] \no \\
&=& \frac{b}{q_{d}-\varepsilon}(1+o(1)) - \frac{\mathbb{E}_{\lambda}^{\nu}[S_{1}^{\lambda-1}|\tau_{s}\geq \lambda]}{q_{d}-\varepsilon}. \no
\end{eqnarray}
Therefore,
\begin{eqnarray}
&&\hspace{-10mm} \mathbb{E}_{\pi}^{\nu}[\tau_{s}-\Lambda|\tau_{s}\geq \Lambda] \no\\
&=& \frac{1}{P_{\pi}^{\nu}(\tau_{s}\geq \Lambda)}\mathbb{E}_{\pi}^{\nu}[\tau_{s}-\Lambda; \tau_{s}\geq \Lambda] \no\\
&=&\frac{1}{P_{\pi}^{\nu}(\tau_{s}\geq \Lambda)} \sum_{\lambda=1}^{\infty}P(\Lambda=\lambda)\mathbb{E}_{\lambda}^{\nu}[\tau_{s}-\lambda|\tau_{s}\geq \lambda]P_{\lambda}^{\nu}(\tau_{s}\geq \lambda) \no\\
&\leq& \frac{b}{q_{d}-\varepsilon} - \frac{\mathbb{E}_{\pi}^{\nu}\left[S_{1}^{\Lambda-1}|\tau_{s}\geq \Lambda\right]}{q_{d}-\varepsilon} + \mathbb{E}_{\pi}^{\nu}[\tilde{T}^{(\Lambda)}_{\varepsilon}|\tau_{s}\geq \Lambda]\no\\
&=& \frac{b}{q_{d}-\varepsilon}(1+o(1)) - \frac{\mathbb{E}_{\pi}^{\nu}\left[S_{1}^{\Lambda-1}|\tau_{s}\geq \Lambda\right]}{q_{d}-\varepsilon}.  \label{eq:result}
\end{eqnarray}
In the following, we show that $\mathbb{E}_{\pi}^{\nu}[S_{1}^{\Lambda-1}|\tau_{s}\geq \Lambda]$ is finite. Let $\tilde{r}$ be the number of nonzero elements in $\{\mu_{1}, \ldots, \mu_{\lambda-1}\}$, and denote $\{b_{1}, \ldots, b_{\tilde{r}}\}$ as the time slots that the observer takes observation before $\lambda$, we have
\begin{eqnarray}
\mathbb{E}_{\lambda}^{\nu}\left[S_{1}^{\lambda-1}\right] &\overset{(a)}=& \mathbb{E}_{\infty}^{\nu}\left[S_{1}^{\lambda-1}\right] \no\\
&=& \mathbb{E}_{\infty}^{\nu}\left[\sum_{i=1}^{\lambda-1}l(Z_{i}) \right] + (\lambda-1)|\log(1-\rho)| \no\\
&=& \mathbb{E}_{\infty}\left[\sum_{i=1}^{\tilde{r}}l(X_{b_{i}}) \right] + (\lambda-1)|\log(1-\rho)| \no\\
&=& -\tilde{r}D(f_{0}||f_{1}) + (\lambda-1)|\log(1-\rho)|, \no
\end{eqnarray}
where (a) is true because $P_{\infty}^{\nu}$ and $P_{\lambda}^{\nu}$ are the same for observations taken before $\lambda$.
Since $\tilde{r}<\lambda$ and $D(f_{0}||f_{1}) \geq 0$, we have
\begin{eqnarray}
-\lambda D(f_{0}||f_{1}) < \mathbb{E}_{\lambda}^{\nu}\left[S_{1}^{\lambda-1}\right] < \lambda|\log(1-\rho)|. \no
\end{eqnarray}
Since
\begin{eqnarray}
\mathbb{E}_{\pi}^{\nu}[S_{1}^{\Lambda-1}] = \sum_{k=1}^{\infty} \mathbb{E}_{\lambda}^{\nu}\left[S_{1}^{\lambda-1}\right] P(\Lambda=\lambda), \no
\end{eqnarray}
we have
\begin{eqnarray}
-\frac{D(f_{0}||f_{1})}{1-\rho} < \mathbb{E}_{\pi}^{\nu}\left[S_{1}^{\Lambda-1}\right] < \frac{|\log(1-\rho)|}{1-\rho}. \no
\end{eqnarray}
Therefore, $\mathbb{E}_{\pi}^{\nu}[S_{1}^{\lambda-1}]$ is bounded. We notice that as $\alpha \rightarrow 0$, $\{\tau_{s} \geq \Lambda\}$ approaches to an almost sure event. Then
$$\mathbb{E}_{\pi}^{\nu}\left[S_{1}^{\Lambda-1}|\tau_{s}\geq\Lambda\right] \rightarrow \mathbb{E}_{\pi}^{\nu}\left[S_{1}^{\Lambda-1}\right] \text{ as } \alpha \rightarrow 0.$$
By \eqref{eq:result} we obtain
\begin{eqnarray}
\mathbb{E}_{\pi}^{\nu}[\tau_{s}-\Lambda|\tau_{s}\geq \Lambda] \leq \frac{b}{q_{d}-\varepsilon}(1+o(1)).
\end{eqnarray}
Since the above equation holds for any $\varepsilon > 0$, then
\begin{eqnarray}
\mathbb{E}_{\pi}^{\nu}[\tau_{s}-\Lambda|\tau_{s}\geq \Lambda] \leq \frac{b}{q_{d}}(1+o(1)). \no
\end{eqnarray}
\end{proof}

Using the above proposition and the fact $\tilde{\tau}^{*} \leq \tau_{s}$, we have
\begin{eqnarray}
\mathbb{E}_{\pi}^{\nu}\left[(\tilde{\tau}^{*} - \Lambda)^{+}\right] &\leq& \mathbb{E}_{\pi}^{\nu}\left[(\tau_{s} - \Lambda)^{+}\right] \no\\
&=& \mathbb{E}_{\pi}^{\nu}[\tau_{s} - \Lambda|\tau_{s} \geq \Lambda][1-P(\tau_{s} < \Lambda)] \no\\
&\leq& \frac{b}{q_{d}}(1-\alpha)(1+o(1)) \no\\
&=& \frac{b}{q_{d}}(1+o(1)). \no
\end{eqnarray}
\bibliographystyle{ieeetr}{}
\bibliography{macros,detection,energyharvester,sensornetwork}

\vspace{-10mm}
\begin{IEEEbiography}[{\includegraphics[width=1in,height=1.25in,clip,keepaspectratio]{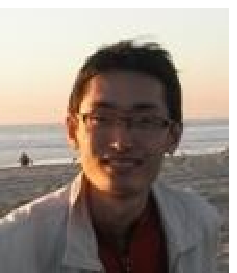}}]{Jun Geng}
(S'13) received the B.E. and M. E. degrees from Harbin Institute of Technology, Harbin, China in 2007 and 2009 respectively. He is currently working towards his Ph.D. degree in the Department of Electrical and Computer Engineering, Worcester Polytechnic Institute.

His research interests include sequential statistical methods, stochastic signal processing, and
their applications in wireless sensor networks and wireless communications.
\end{IEEEbiography}

\vspace{-10mm}
\begin{IEEEbiography}[{\includegraphics[width=1in,height=1.25in,clip,keepaspectratio]{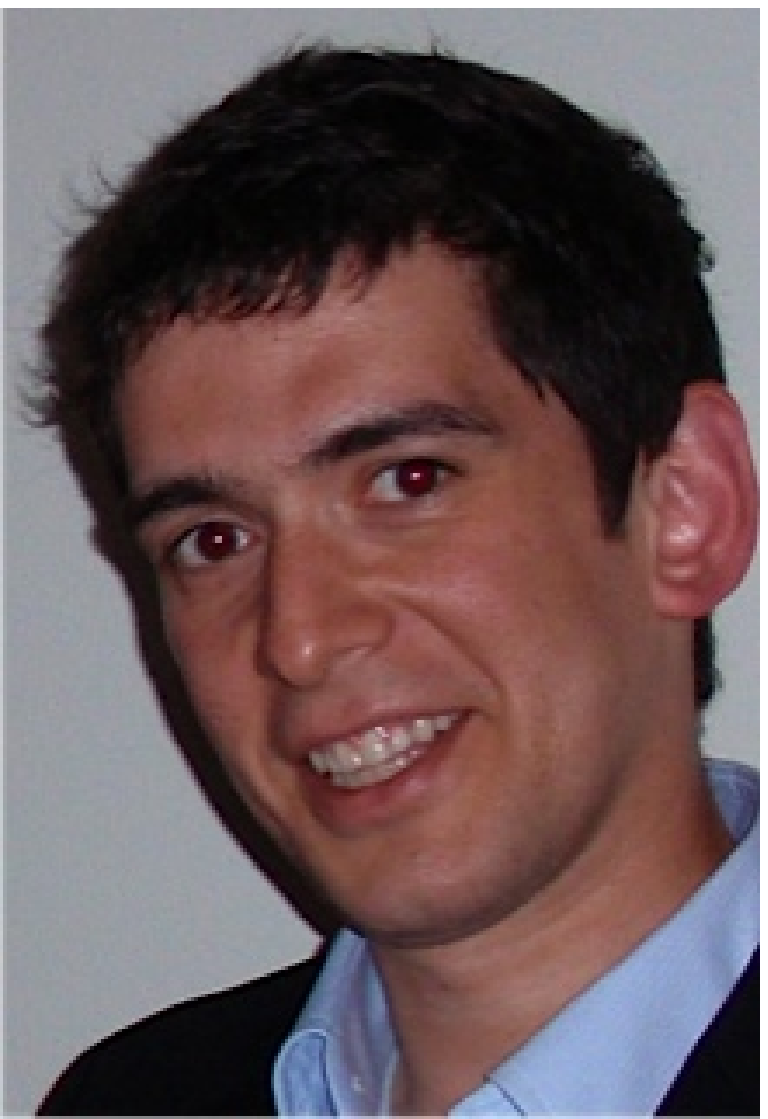}}]{Erhan Bayraktar}
is a professor of Mathematics at the University of Michigan, where he has been since 2004. He is also the holder of the Susan Smith Chair since 2010. Professor Bayraktar's research is in stochastic analysis, control, probability and mathematical finance. He is in the editorial boards of the SIAM Journal on Control and Optimization, Mathematics of Operations Research, and Mathematical Finance. His research is funded by the National Science Foundation. In particular, he received a CAREER grant in 2010.

Professor Bayraktar received his Bachelor's degree (double major in Electrical Engineering and Mathematics) from Middle East Technical University in Ankara in 2000. He received his Ph.D. degree from Princeton in 2004. 
\end{IEEEbiography}

\vspace{-10mm}

\begin{IEEEbiography}[{\includegraphics[width=1in,height=1.25in,clip,keepaspectratio]{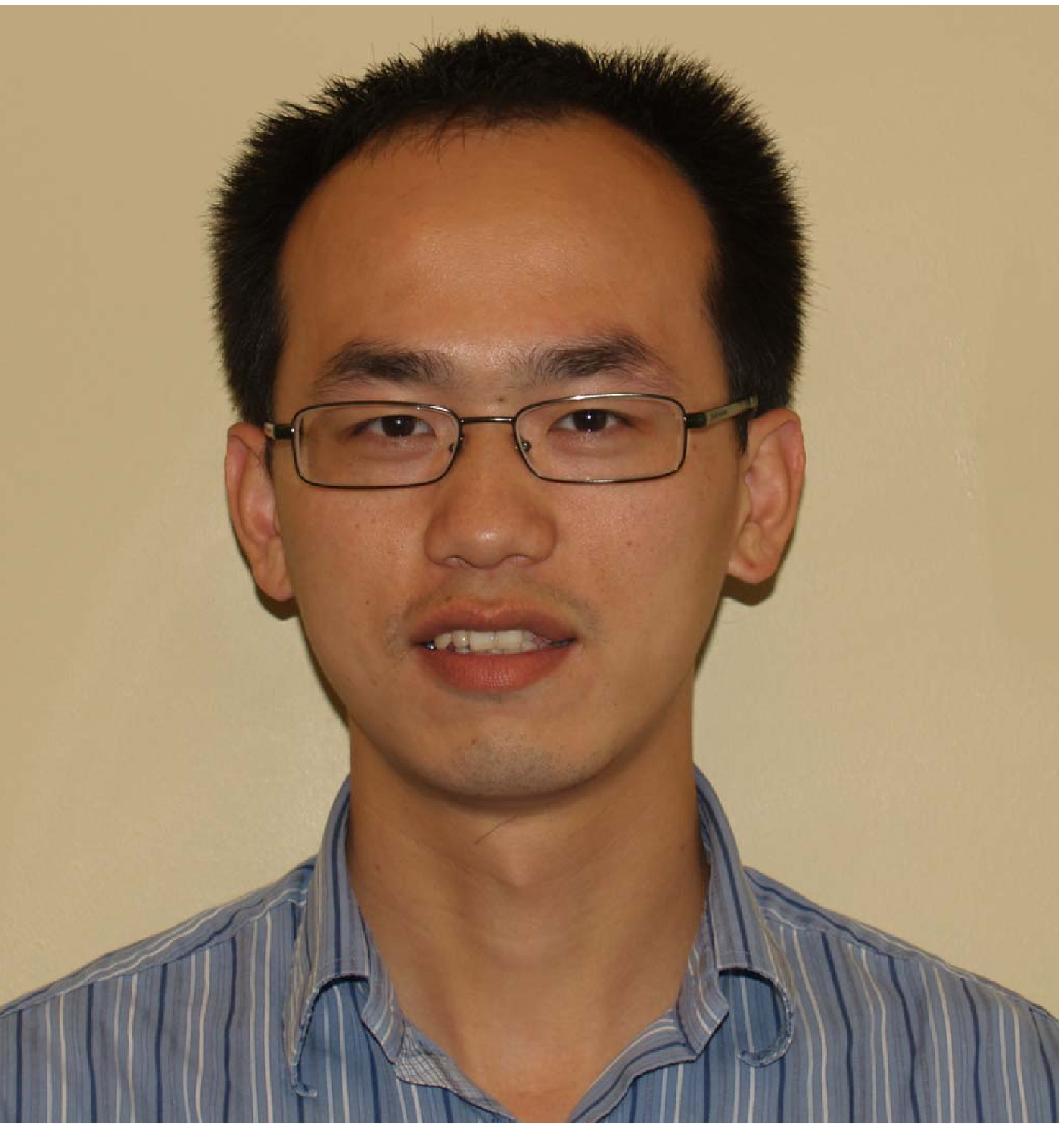}}]{Lifeng Lai}
(M'07) received the B.E. and M. E. degrees from Zhejiang University, Hangzhou, China in 2001 and 2004 respectively, and the PhD degree from The Ohio State University at Columbus, OH, in 2007. He was a postdoctoral research associate at Princeton University from 2007 to 2009, and was an assistant professor at University of Arkansas, Little Rock from 2009 to 2012. Since Aug. 2012, he has been an assistant professor at Worcester Polytechnic Institute. Dr. Lai?s research interests include information theory, stochastic signal processing and their applications in wireless communications, security and other related areas.

Dr. Lai was a Distinguished University Fellow of the Ohio State University from 2004 to 2007. He is a co-recipient of the Best Paper Award from IEEE Global Communications Conference (Globecom) in 2008, the Best Paper Award from IEEE Conference on Communications (ICC) in 2011 and the Best Paper Award from IEEE Smart Grid Communications (SmartGridComm) in 2012. He received the National Science Foundation CAREER Award in 2011, and Northrop Young Researcher Award in 2012. He served as a Guest Editor for IEEE Journal on Selected Areas in Communications, Special Issue on Signal Processing Techniques for Wireless Physical Layer Security. He is currently serving as an Editor for IEEE Transactions on Wireless Communications.
\end{IEEEbiography}
\end{document}